%% file: ms.tex
\documentclass[11pt]{article}
\pdfoutput=1
\input{macros}

\usepackage[utf8]{inputenc}
\usepackage[letterpaper, margin=1in]{geometry}
\usepackage{hyperref}
\usepackage{bbm}
\usepackage{multicol}
\usepackage[shortlabels]{enumitem}
\allowdisplaybreaks

\title{Planted Models for the Densest $k$-Subgraph Problem}
\author{Yash Khanna \\IISc, Bangalore \footnote{Department of Computer Science and Automation, Indian Institute of Science, Bangalore, India.}\\\href{mailto:yashkhanna@iisc.ac.in}{yashkhanna@iisc.ac.in}
\and Anand Louis \\IISc, Bangalore \footnotemark[1]\\\href{mailto:anandl@iisc.ac.in}{anandl@iisc.ac.in}}
\date{\today}

\newcommand{\dksprob}{{\sc Densest $k$-subgraph} problem}
\newcommand{\sseprob}{{\sc Small Set Expansion} problem}
\newcommand{\probabilityconstant}{\kappa}
\newcommand{\matrixconstant}{\xi}

\newcommand{\dks}{{\sc D$k$S}}
\newcommand{\dksparams}{\dks{\sf Exp}$(n, k, d, \delta, d', \lambda)$}
\newcommand{\dksregparams}{\dks{\sf ExpReg}$(n, k, d, \delta, d', \lambda)$}
\newcommand{\dkssparams}{\dks$(n, k, d, \delta, \gamma)$}
\newcommand{\dkssregparams}{\dks{\sf Reg}$(n, k, d, \delta, \gamma)$}

\begin{document}

\maketitle
\begin{abstract}
Given an undirected graph $G$, the \dksprob~(DkS) asks to compute a set 
$S \subset V$ of cardinality $\Abs{S} \leq k$ such that the weight of edges inside $S$ is maximized. This is a fundamental NP-hard problem whose approximability, inspite of many decades of research, is yet to be settled. The current best known approximation algorithm due to Bhaskara et al. (2010) computes a $\bigo{n^{1/4 + \e}}$ approximation in time $n^{\bigo{\ffrac{1}{\e}}}$, for any $\e > 0$.

We ask what are some ``easier'' instances of this problem? We propose some natural semi-random models of instances with a planted dense subgraph, and study approximation algorithms for computing the densest subgraph in them. These models are inspired by the semi-random models of instances studied for various other graph problems such as the independent set problem, graph partitioning problems etc. For a large range of parameters of these models, we get significantly better approximation factors for the \dksprob. Moreover, our algorithm recovers a large part of the planted solution.
\end{abstract}

\input{introduction}
\input{overview}
\input{algorithm_analysis}
\input{other_models}

\paragraph{Acknowledgements.}
We thank Rakesh Venkat for helpful discussions. We also thank the anonymous reviewers for their suggestions and comments on earlier versions of this paper. AL was supported in part by SERB Award ECR/2017/003296 and a Pratiksha Trust Young Investigator Award.

\bibliographystyle{alpha}
\bibliography{references}

\appendix
\input{appendix_2}
\input{appendix}

\end{document}

%% file: macros.tex
\usepackage{amsmath,amstext,amssymb,amsfonts}
\usepackage{color,graphicx}
\usepackage{tabularx}
\usepackage{verbatim}

\usepackage[dvipsnames]{xcolor}

\usepackage{amsthm}
\usepackage{ifdraft}
\usepackage{algorithm2e}
\usepackage{algorithmic}

\usepackage{nicefrac}

\newcommand{\flatfrac}[2]{#1/#2}
\newcommand{\ffrac}{\flatfrac}

\usepackage{microtype}

\newtheorem{theorem}{Theorem}[section]
\newtheorem*{theorem*}{Theorem}

\newtheorem*{claim*}{Claim}

\newtheorem{proposition}[theorem]{Proposition}
\newtheorem*{proposition*}{Proposition}
\newtheorem{lemma}[theorem]{Lemma}
\newtheorem*{lemma*}{Lemma}
\newtheorem{corollary}[theorem]{Corollary}

\newtheorem*{conjecture*}{Conjecture}

\newtheorem{fact}[theorem]{Fact}
\newtheorem*{fact*}{Fact}

\newtheorem*{hypothesis*}{Hypothesis}

\theoremstyle{definition}
\newtheorem{definition}[theorem]{Definition}

\newtheorem{SDP}[theorem]{SDP}

\newtheorem{LP}[theorem]{LP}

\newtheorem{remark}[theorem]{Remark}

\usepackage{prettyref}
\newcommand{\savehyperref}[2]{\texorpdfstring{\hyperref[#1]{#2}}{#2}}

\newrefformat{eq}{\savehyperref{#1}{\textup{(\ref*{#1})}}}
\newrefformat{lem}{\savehyperref{#1}{Lemma~\ref*{#1}}}
\newrefformat{def}{\savehyperref{#1}{Definition~\ref*{#1}}}
\newrefformat{thm}{\savehyperref{#1}{Theorem~\ref*{#1}}}
\newrefformat{cor}{\savehyperref{#1}{Corollary~\ref*{#1}}}
\newrefformat{cha}{\savehyperref{#1}{Chapter~\ref*{#1}}}
\newrefformat{sec}{\savehyperref{#1}{Section~\ref*{#1}}}
\newrefformat{app}{\savehyperref{#1}{Appendix~\ref*{#1}}}
\newrefformat{tab}{\savehyperref{#1}{Table~\ref*{#1}}}
\newrefformat{fig}{\savehyperref{#1}{Figure~\ref*{#1}}}
\newrefformat{hyp}{\savehyperref{#1}{Hypothesis~\ref*{#1}}}
\newrefformat{alg}{\savehyperref{#1}{Algorithm~\ref*{#1}}}
\newrefformat{sdp}{\savehyperref{#1}{SDP~\ref*{#1}}}
\newrefformat{vp}{\savehyperref{#1}{VP~\ref*{#1}}}
\newrefformat{lp}{\savehyperref{#1}{LP~\ref*{#1}}}
\newrefformat{rem}{\savehyperref{#1}{Remark~\ref*{#1}}}
\newrefformat{item}{\savehyperref{#1}{Item~\ref*{#1}}}
\newrefformat{step}{\savehyperref{#1}{step~\ref*{#1}}}
\newrefformat{conj}{\savehyperref{#1}{Conjecture~\ref*{#1}}}
\newrefformat{fact}{\savehyperref{#1}{Fact~\ref*{#1}}}
\newrefformat{prop}{\savehyperref{#1}{Proposition~\ref*{#1}}}
\newrefformat{claim}{\savehyperref{#1}{Claim~\ref*{#1}}}
\newrefformat{relax}{\savehyperref{#1}{Relaxation~\ref*{#1}}}
\newrefformat{red}{\savehyperref{#1}{Reduction~\ref*{#1}}}
\newrefformat{part}{\savehyperref{#1}{Part~\ref*{#1}}}
\newrefformat{prob}{\savehyperref{#1}{Problem~\ref*{#1}}}
\newrefformat{ass}{\savehyperref{#1}{Assumption~\ref*{#1}}}

\newcommand{\Sref}[1]{\hyperref[#1]{\S\ref*{#1}}}

\renewcommand{\leq}{\leqslant}

\renewcommand{\geq}{\geqslant}

\usepackage{bm}

\usepackage{xspace}

\usepackage[pdftex, pagebackref, colorlinks, linkcolor = blue, filecolor = blue, citecolor = red, urlcolor  = blue]{hyperref}

\usepackage{boxedminipage}

\newcommand{\mper}{\,.}
\newcommand{\mcom}{\,,}

\newcommand{\paren}[1]{\left(#1 \right )}

\newcommand{\brac}[1]{[#1 ]}

\newcommand{\set}[1]{\left\{#1\right\}}

\newcommand{\abs}[1]{\left\lvert#1\right\rvert}
\newcommand{\Abs}[1]{\left\lvert#1\right\rvert}

\newcommand{\norm}[1]{\left\lVert#1\right\rVert}

\newcommand{\defeq}{\stackrel{\textup{def}}{=}}

\newcommand{\inprod}[1]{\left\langle #1\right\rangle}

\newcommand{\normt}[1]{\norm{#1}_{\scriptstyle 2}}

\newcommand{\Z}{{\mathbb Z}}

\newcommand{\R}{\mathbb R}

\newcommand{\OPT}{{\sf OPT}}
\newcommand{\ALG}{{\sf ALG}}

\newcommand{\Esymb}{\mathbb{E}}
\newcommand{\Psymb}{\mathbb{P}}

\DeclareMathOperator*{\E}{\Esymb}

\DeclareMathOperator*{\Bern}{{\sf Bern}}
\DeclareMathOperator*{\ProbOp}{\Psymb}

\newcommand{\e}{\epsilon}

\newcommand{\one}{\mathbbm{1}}

\definecolor{DSgray}{cmyk}{0,0,0,0.7}

\let\e\varepsilon

\newcommand{\Erdos}{Erd\H{o}s\xspace}
\newcommand{\Renyi}{R\'enyi\xspace}

\newcommand{\etal}{et al. }
\newcommand{\bigO}{\mathcal{O}}
\newcommand{\bigo}[1]{\bigO\left(#1\right)}

\newcommand{\argmax}{{\sf argmax}}



%% file: introduction.tex
\section{Introduction}
Given a weighted undirected graph $G = (V,E,w)$ with non-negative edge weights given by $w:E \to \R^+$, and an integer $k \in \Z^+$, the \dksprob~(\dks) asks to compute a set $S \subset V$ of cardinality $\Abs{S} \leq k$ such that the weight of edges inside $S$ (i.e., $\sum_{i,j \in S} w\paren{\set{i,j}}$) is maximized (if $\set{i,j} \notin E$, we assume w.l.o.g. that $w\paren{\set{i,j}} = 0$). Computing the \dks~of a graph is a fundamental NP-hard problem. There has been a lot of work on studying approximation algorithms for \dks, we give a brief survey in \prettyref{sec:related}.

The current best known approximation algorithm \cite{DBLP:conf/stoc/BhaskaraCCFV10} 
computes an $\bigo{n^{1/4 + \e}}$ approximation in time $n^{\bigo{\ffrac{1}{\e}}}$ for any $\e > 0$. On the hardness side, Manurangsi \cite{DBLP:conf/stoc/Manurangsi17} showed that assuming the exponential time hypothesis (ETH), there is no polynomial time algorithm that approximates this to within $n^{1/(\log \log n)^c}$ factor where $c >0$ is some fixed constant. There are hardness of approximation results known for this problem assuming various other hardness assumptions, see \prettyref{sec:related} for a brief survey. But there is still a huge gap between the upper and lower bounds on the approximability of this problem.

Given this status of the approximability of the \dksprob, we ask what are some ``easier'' instances of this problem? We propose some natural semi-random models of instances with a planted dense subgraph, and study approximation algorithms for computing the densest subgraph in them. Studying semi-random models of instances has been a very fruitful direction of study towards understanding the complexity for various NP-hard problems such as graph partitioning problems \cite{Makarychev:2012:AAS:2213977.2214013, Makarychev:2014:CFA:2591796.2591841, DBLP:conf/icalp/LouisV18, DBLP:conf/fsttcs/LouisV19}, independent sets \cite{10.1006/jcss.2001.1773, DBLP:conf/soda/McKenzieMT20}, graph coloring \cite{Alon:1997:STC:270560.270935, Coja-Oghlan:2007:CSG:1273845.1273847, David:2016:ERP:2897518.2897561}, etc. By studying algorithms for instances where some parts are chosen to be arbitrary and some parts are chosen to be random, one can understand which aspects of the problem make it computationally intractable. Besides being of natural theoretical interest, studying approximation algorithms for semi-random models of instances can also be practically useful since some natural semi-random models of instances  can be better models of instances arising in practice than the worst-case instances. Therefore, designing algorithms specifically for such models can help to bridge the gap between theory and practice in the study of algorithms. Some random and semi-random models of instances of the \dksprob~(and its many variants) have been studied in \cite{959929, DBLP:conf/stoc/BhaskaraCCFV10, 2014arXiv1406.6625H, 10.1007/s10957-015-0777-x, Montanari, 7523889, e86c6c373bf9484a861530abef1cb4a5, bombina2019convex}, we discuss them in \prettyref{sec:related}. Our models are primarily inspired by the densest subgraph models mentioned above as well as the semi-random models of instances for other problems \cite{10.1006/jcss.2001.1773, DBLP:conf/soda/McKenzieMT20} studied in the literature. For a large range of parameters of these models, we get significantly better approximation factors for the \dksprob, and also show that we can recover a large part of the planted solution. 

We note that semidefinite programming (SDP) based methods have been popularly used in many randomized models for different problems, including the \dksprob~\cite{2014arXiv1406.6625H, 7523889, e86c6c373bf9484a861530abef1cb4a5}. And thus, another motivation for our work is to understand the power of SDPs in approximating the \dksprob. Since even strong SDP relaxations of the problem have a large integrality gap \cite{10.5555/2095116.2095150} for worst case instances (see \prettyref{sec:related}), we ask what families of instances can SDPs approximate well? 
In addition to being of theoretical interest, algorithms using the basic SDP also have a smaller running time. In comparison, the algorithm of  \cite{DBLP:conf/stoc/BhaskaraCCFV10} produces an $\bigo{n^{1/4 + \e}}$ approximation for worst-case instances in time $n^{\bigo{\ffrac{1}{\e}}}$; their algorithm is based on rounding an LP hierarchy, but they also show that their algorithm can be executed without solving an LP and obtain the same guarantees.

\subsection{Our models and results} 
\label{sec:results}

The main inspiration for our models are the semi-random models of instances for the independent set problem
\cite{10.1006/jcss.2001.1773, DBLP:conf/soda/McKenzieMT20}.
Their instances are constructed as follows. Starting with a set of vertices $V$, a subset of 
$k$ vertices is chosen to form the independent set $S$, and edges are added between each pair 
in $S \times (V \setminus S)$ independently with probability $p$. Finally, 
an arbitrary graph is added on $V \setminus S$. 
They study the values of $k$ and $p$ for which they can recover a large independent set. 
Our models can be viewed as analogs of this model to the \dksprob: 
edges are added between each pair in $S \times (V \setminus S)$ independently with probability $p$,
and then edges are added in $S$ to form a dense subset. Since 
we also guarantee that we can recover a large part of the planted dense subgraph $S$, 
we also need to assume that the graph induced on $V \setminus S$ is ``far'' from
containing a dense subgraph. 
We now define our models.

\begin{definition}[\dksparams]
\label{def:dks-expander}
An instance of \dksparams~is generated as follows,
\begin{enumerate}
\item  
\label{step:one}
We partition $V$ into two sets, $S$ and $V \setminus S$ with $\abs{S} = k$. We add edges (of weight 1) between pairs in $S \times \paren{V\setminus S}$ independently with probability $p \defeq \ffrac{\delta d}{k}$.
\item 
\label{step:two} 
We add edges of arbitrary non-negative weights between arbitrary pairs of vertices in $S$ such that the graph induced on $S$ has average weighted degree $d$.
\item 
\label{step:three}
We add edges of arbitrary non-negative weights between arbitrary pairs of vertices in $V \setminus S$ such that the graph induced on $V \setminus S$ is a $(d',\lambda)$-expander (see \prettyref{def:expander} for definition). 
\item 
\label{step:four} (\emph{Monotone adversary}) Arbitrarily delete any of the edges added in \prettyref{step:one} and \prettyref{step:three}. 
\item 
\label{step:five} Output the resulting graph.
\end{enumerate}
\end{definition}

We note that the \prettyref{step:two}, \prettyref{step:three}, and \prettyref{step:four} in the construction of the instance above are adversarial steps.

\dksparams~are a class of instances that have a prominent dense subset of size $k$. Note that, since the graph induced on $V \setminus S$ is a subset of an expander graph, it would not have any dense subsets. We also note that the monotone adversary can make significant changes to graph structure. For example, the graph induced on $V \setminus S$ can be neither $d'$-regular nor an expander after the action of the monotone adversary. 

We require $\delta < 1$ in \prettyref{step:one} for the following reason. For any fixed set $S' \subset V \setminus S$ such that $\Abs{S'} = \bigo{k}$, 
the expected weight of edges in the bipartite graph induced on $S \cup S'$ is $\bigo{\delta k d}$. Since we want the graph induced on $S$ to be the densest $k$-subgraph (the total of edges in the graph induced on $S$ is $kd/2$), we restrict $\delta$ to be at most $1$.

We present our main results below, note that our algorithm outputs a dense subgraph of size $k$ and its performance is measured with respect to the density of the planted subgraph $G\brac{S}$, i.e. $\ffrac{kd}{2}$. 

\begin{definition}
	We define $\rho(V') \defeq \ffrac{\paren{ \sum_{i,j \in V'} w \paren{\set{i,j}}}}{2}$ for any $V' \subseteq V$.
\end{definition}

\begin{theorem}[Informal version of \prettyref{thm:main_exp}]
	\label{thm:inf_main_exp}
	Given an instance of \dksparams~\\where 
	\[  
	\delta = \Theta\paren{\dfrac{kd'}{nd}},\qquad \dfrac{\delta d}{k} = \Omega\paren{\dfrac{\log n}{n}},\qquad \text{ and } \qquad
	\nu = \Theta\paren{\sqrt{\delta+\dfrac{\lambda+\sqrt{d'}}{d}}}, \] 
	there exists a deterministic polynomial time algorithm that outputs
	with high probability (over the instance) a vertex set $\mathcal{Q}$ of 
	size $k$ such that 
	$ \rho\paren{ \mathcal{Q}} \geq \paren{1-\nu} \dfrac{kd}{2} \mper$
	The above algorithm also computes a vertex set $T$ such that
	\begin{multicols}{2}
		\begin{enumerate}[(a)]
			\item $\abs{T} \leq \paren{1 + \bigo{\nu}}k \mper$
			\item $\rho(T \cap S) \geq \paren{1-\bigo{\nu}} \dfrac{kd}{2}\mper$
		\end{enumerate}
	\end{multicols}
\end{theorem} 

\begin{remark}
	\label{rem:simplified_apx_1}
	In \prettyref{thm:inf_main_exp}, we restrict the range of $\delta$ for the following reason. 
	An interesting setting of parameters is when the average degree of vertices in $S$ and 
	$V \setminus S$ are within constant factors of each other. Then the expected average degree of a vertex in $S$ is $d + p(n-k)$. And for a vertex in $V \setminus S$, the expected average degree is $d' + kp$. Thus setting,
	\[
		d+p(n-k) = \Theta(d'+kp) \implies \delta = \Theta\paren{\dfrac{kd'}{nd}} \qquad\paren{\text{Recall, } p = \frac{\delta d}{k}}\mper
	\]
\end{remark}

We also study another interesting model with a different assumption on the subgraph $G\brac{V \setminus S}$. 
\begin{definition}
\dkssparams~is generated similarly to \dksparams~except in \prettyref{step:three}, where 
we add edges between arbitrary pairs of vertices in $V \setminus S$ such that the graph induced on $V \setminus S$ has the following property : 
$ \rho(V') \leq \gamma d \Abs{V'} \quad \forall V' \subseteq V \setminus S \mper$
\end{definition}

By construction, the graph induced on $V \setminus S$ does not have very dense subsets. 

\begin{theorem}[Informal version of \prettyref{thm:main_gamma}]
	\label{thm:inf_main_gamma}
	Given an instance of \dkssparams~where 
	\[  
	\delta = \Theta\paren{\dfrac{k}{n}},\qquad \dfrac{\delta d}{k} = \Omega\paren{\dfrac{\log n}{n}},\qquad \text{ and }
	\qquad\tau = \Theta\paren{\sqrt{\delta+\gamma+\dfrac{1}{\sqrt{d}}}}, \] 
	there is a deterministic polynomial time algorithm that outputs
	with high probability (over the instance) a vertex set $\mathcal{Q}$ of 
	size $k$ such that 
	$ \rho\paren{ \mathcal{Q}} \geq \paren{1-\tau} \dfrac{kd}{2}\mper$
	The above algorithm also computes a vertex set $T$ such that
	\begin{multicols}{2}
		\begin{enumerate}[(a)]
			\item $\abs{T} \leq \paren{1+ \bigo{\tau}} k \mper$
			\item $\rho(T \cap S) \geq \paren{1-\bigo{\tau}} \dfrac{kd}{2}\mper$
		\end{enumerate}
	\end{multicols}
\end{theorem} 

\subsubsection*{Other results}
We also study two variants of \dksparams~and \dkssparams~where the subgraph $G\brac{S}$ is $d$-regular.

\begin{enumerate}
\item \dksregparams~is same as \dksparams~except in \prettyref{step:two}, which requires the subgraph $G\brac{S}$ to be an arbitrary $d-$regular graph.
\begin{theorem}[Informal version of \prettyref{thm:main_exp_reg}]
	\label{thm:inf_main_exp_reg}
	Given an instance of \dksregparams\\ where 
	\[  
	\delta = \Theta\paren{\dfrac{kd'}{nd}},\qquad \dfrac{\delta d}{k} = \Omega\paren{\dfrac{\log n}{n}},\qquad \text{ and } \qquad
	\nu' = \Theta\paren{\dfrac{\sqrt{d'}}{d\left(1 - \delta - \dfrac{\lambda}{d}\right)}}, \] 
	there is a deterministic polynomial time algorithm that outputs
	with high probability (over the instance) a vertex set $\mathcal{Q}$ of 
	size $k$ such that 
	\begin{multicols}{2}
		\begin{enumerate}
			\item $\rho\paren{ \mathcal{Q}} \geq \paren{1-\nu'} \dfrac{kd}{2} \mper$
			\item $\abs{\mathcal{Q} \cap S} \geq \paren{1-\bigo{\nu'}} k \mper$
		\end{enumerate}
	\end{multicols}
\end{theorem}

\item \dkssregparams~is same as \dkssparams~except in \prettyref{step:two}, which requires the subgraph $G\brac{S}$ to be an arbitrary $d-$regular graph.
\begin{theorem}[Informal version of \prettyref{thm:main_gamma_reg}]
	\label{thm:inf_main_gamma_reg}
	Given an instance of \dkssregparams\\ where 
	\[  
\delta = \Theta\paren{\dfrac{k}{n}},\qquad \dfrac{\delta d}{k} = \Omega\paren{\dfrac{\log n}{n}},\qquad \text{ and }
\qquad\tau' = \Theta\paren{\dfrac{1}{\sqrt{d}\paren{1-\gamma-\delta}}}, \]
there is a deterministic polynomial time algorithm that outputs
	with high probability (over the instance) a vertex set $\mathcal{Q}$ of 
	size $k$ such that 
	\begin{multicols}{2}
		\begin{enumerate}
			\item $\rho\paren{ \mathcal{Q}} \geq \paren{1-\tau'} \dfrac{kd}{2} \mper$
			\item $\abs{\mathcal{Q} \cap S} \geq \paren{1-\bigo{\tau'}} k \mper$
		\end{enumerate}
\end{multicols}
\end{theorem}
\end{enumerate}

We will show that for most natural regime of parameters, we get better approximation factors in the case when $G\brac{S}$ is a $d$-regular graph. We formalize this in \prettyref{app:b}.

\begin{remark}
	It has been pointed out to us by anonymous reviewers that for a large range of parameters of the \dkssparams~and \dkssregparams~models, $\argmax_{W \subseteq V} \rho(W)/\abs{W}$ will be a subset of $S$; for any graph $G=(V, E)$, the algorithm due to Charikar \cite{Charikar:2000:GAA:646688.702972} can be used to compute $\argmax_{W \subseteq V} \rho(W)/\abs{W}$ in polynomial time. It is plausible that using this algorithm iteratively, one can recover a ``large'' part of $S$. However the algorithm described in \prettyref{thm:inf_main_gamma} and \prettyref{thm:inf_main_gamma_reg} gives a more direct approach to recover a large part of $S$.
\end{remark}

\subsection{Notation}
\label{sec:notation}
We use $n \defeq \abs{V}$, and use $V$ and $[n] \defeq \set{1,2,\hdots,n}$ interchangeably.
We assume w.l.o.g. that $G$ is a complete graph: if $\set{i,j} \notin E$, we add $\set{i,j}$
to $E$ and set $w \paren{\set{i,j}} = 0$.
We use $A$ to denote the weighted adjacency matrix of $G$, i.e. 
$A_{ij} = w\paren{\set{i,j}} \ \forall i,j \in V$. 
The degree of vertex $i$ is defined as $d_i \defeq\sum\limits_{j \in V} w\paren{\set{i,j}}$. 

For $V' \subseteq V$, we use $G\brac{V'}$ to denote the subgraph induced on $V'$ and $\overline{V'}$ to denote $V \setminus V'$. 
For a vector $v$, we use $\norm{v}$ to denote the $\normt{v}$.
For a matrix $A$, we use $\norm{A}$ to denote the spectral norm $\norm{A} \defeq \max\limits_{x \neq 0} \dfrac{\norm{Ax}}{\norm{x}}$.

We define probability distributions $\mu$ over finite sets $\Omega$. For a random variable (r.v.) $X : \Omega \to \R$, its expectation is denoted by $\E_{\omega \sim \mu}\brac{X}$. In particular, we define the two distributions which we use below.
\begin{enumerate}
	\item For a vertex set $V'\subseteq V$, we define a probability (uniform) distribution $(f_{V'})$ on the vertex set $V'$ as follows. For a vertex $i \in V'$, $f_{V'}(i) = \dfrac{1}{\abs{V'}}$. We use $i \sim V'$ to denote $i \sim f_{V'}$ for clarity.
	\item For a vertex set $V'\subseteq V$, we define a probability distribution $(f_{E(G[V'])})$ on the edges of $G[V']$ as follows. For an edge $e \in E(G[V'])$, $f_{E(G[V'])}(e) = \dfrac{w\paren{e}}{\rho(V')}$. Again, we use $e \sim E(G[V'])$ to denote $e \sim f_{E(G[V'])}$ for convenience.
\end{enumerate}

\begin{definition}[$(d,\lambda)$-expanders]
\label{def:expander}
A graph $H = (V, E, w)$ is said to be a $(d,\lambda)$-expander if $H$ is $d$-regular and $\Abs{\lambda_i} \leq \lambda$, $\forall i \in [n]\setminus \set{1}$, where $\lambda_1 \geq \lambda_2 \hdots \geq \lambda_{n}$ are the eigenvalues of the weighted adjacency matrix of $H$.
\end{definition}

\subsection{Related Work}
\label{sec:related}

\subparagraph*{Densest $k$-subgraph.}
There has been a lot of work on the \dksprob~and its variants. The current best known approximation algorithm, due to Bhaskara \etal \cite{DBLP:conf/stoc/BhaskaraCCFV10}, gives an approximation ratio of $\bigO(n^{1/4+\e})$ in time $n^{\bigO(1/\e)}$, for all values of $\e > 0$ (for $\e = 1/\log n$, we get a ratio of $\bigo{n^{1/4}}$). They also extend their approach to give a $\bigO(n^{\ffrac{1}{4}-\e})$ approximation algorithm which runs in time $2^{n^{\bigO(\e)}}$. They improved the prior results of Feige \etal \cite{DBLP:journals/algorithmica/FeigePK01} which gave a $n^{1/3-\e}$ approximation for some small $\e > 0$. \cite{DBLP:journals/algorithmica/FeigePK01} also give a greedy algorithm which has an approximation factor of $\bigo{\ffrac{n}{k}}$.

When $k = \Theta(n)$, Asahiro et al. \cite{10.1007/3-540-61422-2_127} gave a constant factor approximation algorithm. Many other works have looked at this problem using linear and semidefinite programming techniques. Srivastav \etal \cite{Srivastav:1998:FDS:646687.702946} gave a randomized rounding algorithm using a SDP relaxation in the case when $k = n/c$ for $c > 1$, they improved the constants for certain values of $k$ over the results of \cite{10.1007/3-540-61422-2_127}. Feige and Langberg \cite{FEIGE2001174} use a different SDP to get an approximation of slightly above $k/n$ for the case when $k$ is roughly $\ffrac{n}{2}$. Feige and Seltser \cite{Feige97onthe} construct examples for which their SDP has an integrality gap of $\Omega{(n^{1/3})}$.

There has been work done on a related problem called the maximum density subgraph, where the objective is to find a subgraph which maximizes the ratio of number of edges to the number of vertices. Goldberg \cite{Goldberg:1984:FMD:894477} and Gallo \etal \cite{Gallo:1989:FPM:63408.63424} had given an algorithm to solve this problem exactly using maximum flow techniques. Later, Charikar \cite{Charikar:2000:GAA:646688.702972} gave an algorithm based on a linear programming method. This paper also solves the problem for directed graphs using a notion of density given by Kannan and Vinay \cite{kannan1999analyzing}. Khuller and Saha \cite{Khuller:2009:FDS:1577399.1577451} gave a max-flow based algorithm in the directed setting. 

On the hardness side, Khot \cite{Khot:2006:ROP:1328008.1328009} showed that it does not have a PTAS unless NP has subexponential algorithms. There has been some works based on some other hardness assumptions. Assuming the small-set expansion hypothesis, Raghavendra and Steurer \cite{Raghavendra:2010:GEU:1806689.1806792} show that it is NP-hard to approximate \dks~to any constant factor. Under the deterministic ETH assumption, Braverman \etal \cite{Braverman:2017:EHD:3039686.3039772} show that it requires $n^{\Omega(\log n)}$ time to approximate \dks~with perfect completeness to within $1+\varepsilon$ factor (for a universal constant $\varepsilon >0$). More recently Manurangsi \cite{DBLP:conf/stoc/Manurangsi17} showed assuming the exponential time hypothesis (ETH), that there is no polynomial time algorithm that approximates this to within $n^{1/(\log \log n)^c}$ factor where $c >0$ is some fixed constant independent of n. 

Bhaskara \etal \cite{10.5555/2095116.2095150} study strong SDP relaxations of the problem and show that the integrality gap of \dks~remains $n^{\Omega_\varepsilon\paren{1}}$ even after $n^{1-\varepsilon}$ rounds of the Lasserre hierarchy. Also for $n^{\Omega\paren{\varepsilon}}$ rounds, the gap is as large as $n^{2/53-\varepsilon}$. Moreover for the Sherali-Adams relaxation, they show a lower bound of $\Omega\paren{\ffrac{n^{1/4}}{\log^3 n}}$ on the integrality gap for $\Omega\paren{\ffrac{\log n}{\log\log n}}$ rounds.

Ames \cite{10.1007/s10957-015-0777-x} studies the planted \dks~problem using a non-SDP convex relaxation for instances of the following kind. Let $S$ be the planted dense subgraph (of size $k$), they claim that if $G\brac{S}$ contains at least ${k \choose 2} - c_1 k^2$ edges and the subgraph $G\brac{V \setminus S}$ contains at most $c_2 k^2$ edges where $c_1, c_2$ are constants depending on other parameters of the graph like the density of the subgraph $G\brac{S}$ etc, then under some mild technical conditions, they show that the unique optimal solution to their convex program is integral and corresponds to the set $S$. They also study analogous models for bipartite graphs.

\subparagraph*{Random models for DkS.}
Bhaskara \etal \cite{DBLP:conf/stoc/BhaskaraCCFV10} study a few random models of instances for the \dksprob, we describe them here. 
Let $\mathcal{D}_1$ denote the distribution of \Erdos-\Renyi random graphs $G(n, p)$ and let  
$\mathcal{D}_2$ denote the distribution of graphs constructed as follows. Starting with a ``host graph'' of average degree $D$ ($D\defeq np$),  a set $S$ of $k$ vertices is chosen arbitrarily and the subgraph on $S$ is replaced with a dense subgraph of average degree $d$. 
Given $G_1 \sim \mathcal{D}_1$ and $G_2 \sim \mathcal{D}_2$, the problem is to distinguish between the two distributions. They consider this problem in three different models with varying assumptions on $\mathcal{D}_2$,
(i) \emph{Random Planted Model} : the host graph and the planted dense subgraph are random,
(ii) \emph{Dense in Random Model} : an arbitrary dense graph is planted inside a random graph, and
(iii) \emph{Dense vs Random Model} : an arbitrary dense graph is planted inside an arbitrary graph.

The \emph{planted dense subgraph} recovery problem is similar in spirit to the \emph{Random Planted Model} where the goal is to recover a hidden community of size $k$ within a larger graph which is constructed as follows : two vertices are connected by an edge with probability $p$ if they belong to the same community and with probability $q$ otherwise. The typical setting of parameters is, $p > q$. The works by \cite{959929, 2014arXiv1406.6625H, 10.1007/s10957-015-0777-x, Montanari, 7523889, e86c6c373bf9484a861530abef1cb4a5, bombina2019convex} studies this problem using SDP based, spectral, statistical, message passing algorithms etc.

We give a brief overview of their distinguishing algorithms in the three models. Given a graph on $n$ vertices with average degree $d_{\text{avg}}$, its $\log$-density is defined as $\dfrac{\log d_{\text{avg}}}{\log n}$. Let $\Theta_1$ and $\Theta_2$ denote the $\log$-density of $G_1$ and the $\log$-density of the planted subgraph $G_2[S]$ respectively. Their algorithm is based on the counts of a specially constructed small-sized tree (the size of which is parameterized by relatively prime integers $r,s$ such that $s > r > 0$) as a subgraph in $G_1$ and $G_2$.  They show that if $\Theta_1 \leq r/s$, then $G_1$ will have at most poly-logarithmic $(\bigo{\log n}^{s-r})$ number of such subtrees. On the other hand, when $\Theta_2 \geq r/s+\varepsilon$ where $\e > 0$ is a small constant, they show that there at least $k^{\e}$ such subtrees (even in the \emph{Dense vs Random Model}). Now if $k > (\log n)^{\omega(1)}$, they use this difference in the $\log$-densities to show the gap between counts of such trees in $G_1$ and $G_2$, and hence are able to distinguish between the two distributions. They show that the running time of this algorithm is $n^{\bigo{r}}$. Also for constant $\Theta_1$ and $\Theta_2$, the running time is $n^{\bigo{\ffrac{1}{\paren{\Theta_2-\Theta_1}}}}$ (\cite{DBLP:conf/stoc/BhaskaraCCFV10, adityathesis}). We call this algorithm the ``subgraph counting'' algorithm.

The distinguishing problem can be restated as the following : For a given $n, k, p$, we are interested in finding the smallest value of $d$ for which the problem can be solved. For a certain range of parameters, spectral, SDP based methods, etc. can be used to work for small values of $d$. For example, in the \emph{Dense vs Random Model}, when $k > \sqrt{n}$ a natural SDP relaxation of \dks~can be used to distinguish between $G_1$ and $G_2$ for $d > \sqrt{D}+kD/n$ (which is smaller than  $D^{\log_{n} k}$, the threshold of the subgraph counting algorithm). They upper bound the cost of the optimal SDP solution for a random graph $G_1$, by constructing a feasible dual solution which certifies (w.h.p.) that it cannot contain a $k$-subgraph with density more than that of $\sqrt{D}+kD/n$. We use their results in bounding the cost of the SDP contribution from $G\brac{V \setminus S}$ in the \dksparams~and \dksregparams~models. 

The distribution $\mathcal{D}_2$ of graphs considered in the \emph{Dense in Random Model} (arbitrary dense graph planted in a random graph)  is similar to a subset of \dksparams~instances since $G\brac{S}$ is an arbitrary dense subgraph in both models and $G\brac{S, V \setminus S}$ is a random graph in both the models. The difference is in the subgraph $G\brac{V \setminus S}$, where this is a random graph in the {\em Dense in Random} model whereas our models require it to be a regular expander. While our proofs require the expander to be regular, they can also be made to work for random graphs since we use the bound on the SDP value from \cite{DBLP:conf/stoc/BhaskaraCCFV10} (analysis in \prettyref{sec:2.2}). We note that while random graphs are good expanders w.h.p., the converse of this fact is not true in general, since there are known deterministic constructions of expander graphs.

We look at the range of parameters where the following two algorithms can be used to solve the \emph{Dense in Random} problem. One is the SDP based algorithm proposed in our work (closely related to \dksparams~model) and second is the subgraph counting algorithm which uses the difference in the $\log$-densities of the planted subgraph and the host graph to distinguish the two distributions from \cite{DBLP:conf/stoc/BhaskaraCCFV10, adityathesis}. For the purposes of comparison, we consider the case when $k, d = poly(n)$ and $p = 1/poly(n)$. Also we ignore the low-order terms in these expressions. In this regime, our algorithms' threshold is 
\begin{equation}
	\label{eq:alg_threshold}
	d =  \Omega\paren{\max\set{pk, \sqrt{np}}}
\end{equation}
since we can use the objective value of the \prettyref{sdp:dks} to distinguish between the cases in this range of $d$. For $G_1$, this value is at most $\ffrac{k\paren{pk+\sqrt{np}}}{2}$ (\prettyref{lem:sdp-random-graph}) while for $G_2$ it is at least $\ffrac{kd}{2}$. Moreover, \prettyref{alg:one} can be used to recover a part of the planted solution as the value of $\nu$ is small (when $d$ satisfies \prettyref{eq:alg_threshold}, $\nu$ is bounded away and smaller than 1) in this regime (see \prettyref{sec:2} and \prettyref{thm:main_exp}).

The counting algorithms' threshold (or the $\log$-density threshold) is
\[\dfrac{\log d}{\log k} - \dfrac{\log np}{\log n} > 0 \iff \log d > \dfrac{\log k\log np}{\log n} \iff d = \Omega\paren{(np)^{\log_{n} k}}\]
and its running time is $n^{\bigo{\dfrac{1}{\log_{k} d - \log_{n} np}}}$. We look at different ranges of $k$ and compare the values of $d$ for which the two algorithms can solve the distinguishing problem.

\begin{enumerate}
	\item $k = \Theta\paren{\sqrt{n}}$.\\
		In this case, $\max\set{pk, \sqrt{np}} = \sqrt{np}$. This matches with the $\log$-density threshold. Note that for $p = \Theta\paren{\ffrac{1}{\sqrt{n}}}$, we get $d = \Omega\paren{n^{1/4}}$. To the best of our knowledge, there is no poly-time algorithm which beats this lower bound.
	\item $k = \omega\paren{\sqrt{n}}$.\\
	In this setting, $(np)^{\log_{n} k} = \omega\paren{\sqrt{np}}$. Also, $(np)^{\log_{n} k} = k(p)^{\log_{n} k} = \omega\paren{pk}$. Thus our algorithm has a better threshold in this regime. There is a spectral algorithm, see Section 6.2 of \cite{DBLP:conf/stoc/BhaskaraCCFV10}, which uses the second eigenvalue of the adjacency matrix which can distinguish with the same threshold as our algorithm in this regime.
	\item $k = o\paren{\sqrt{n}}$.\\
	In this case, $(np)^{\log_{n} k} = o\paren{\sqrt{np}}$. Here the $\log$-density threshold is smaller than our threshold. Therefore the algorithm by Bhaskara \etal \cite{DBLP:conf/stoc/BhaskaraCCFV10} works for a larger range of parameters than our algorithms.
\end{enumerate}

\subparagraph*{Other semi-random models.}
Semi-random instances of many other fundamental problems have been studied in the literature.
This includes the unique games problem \cite{DBLP:conf/focs/KollaMM11}, graph coloring \cite{Alon:1997:STC:270560.270935, Coja-Oghlan:2007:CSG:1273845.1273847, David:2016:ERP:2897518.2897561}, graph partitioning problems such as balanced-cut, multi-cut, small set expansion \cite{Makarychev:2012:AAS:2213977.2214013, Makarychev:2014:CFA:2591796.2591841,DBLP:conf/icalp/LouisV18, DBLP:conf/fsttcs/LouisV19}, etc. 
\cite{pmlr-v49-makarychev16} studies the problem of learning communities in the Stochastic Block Model in the presence of adversarial errors. 

McKenzie, Mehta and Trevisan \cite{DBLP:conf/soda/McKenzieMT20} study the complexity of the independent set problem in the Feige-Kilian model \cite{10.1006/jcss.2001.1773}. Instead of using a SDP relaxation for the problem, they use a ``crude'' SDP (introduced in \cite{DBLP:conf/focs/KollaMM11}) which exploits the geometry of vectors (orthogonality etc.) to reveal the planted set. They bound the SDP contribution by the vertex pairs, $S \times V \setminus S$ using the Grothendieck inequality and thereby showing that the vectors in $S$ are ``clustered'' together. Their algorithm outputs w.h.p. a large independent set when $k = \Omega\paren{\ffrac{n^{2/3}}{p^{1/3}}}$. Also, for the parameter range $k = \Omega\paren{\ffrac{n^{2/3}}{p}}$, it outputs a list of at most $n$ independent sets of size $k$, one of which is the planted one.

\subparagraph*{Semi-random models for graph partitioning problems.}
The problem of \dks~is very closely related to the \sseprob~(SSE, henceforth). This problem has been very well studied in the literature. At the first glance, the problem of \dks~can be thought of as finding a small set $S$ of size $k$ which is non-expanding. The densest set is typically a non-expanding set because most of the edges incident on $S$ would remain inside it than leaving it. But the converse is not true, since all sets of cardinality $k$ which have small expansion are not dense. In particular, in our model, by the action of the monotone adversary on $V \setminus S$, there can exist many small sets (of size $\bigo{k}$) which not only have a very small fraction of edges going outside but can have very few edges left inside as well. This makes the problem of \dks~very different from the SSE problem. Nevertheless, we survey some related works of semi-random models of SSE. The works \cite{10.1145/1806689.1806776, DBLP:journals/siamcomp/BansalFKMNNS14} study the worst-case approximation factors for the SSE problem and give bi-criteria approximation algorithms for the same. Their algorithms are also based on rounding a SDP relaxation.

Makarychev, Markarychev and Vijayaraghavan \cite{Makarychev:2012:AAS:2213977.2214013} study the complexity of many graph partitioning problems including balanced cut, SSE, and multi-cut etc. They consider the following model : Partition $V$ into $(S, V \setminus S)$ such that $G\brac{S}$ and $G\brac{V \setminus S}$ are arbitrary while $G\brac{S, V \setminus S}$ is a random graph with some probability $\varepsilon$. They allow an adversary to add edges within $S$ and $V \setminus S$, and delete any edges across these sets. They get constant factor bi-criteria approximation algorithms (under some mild technical conditions) in this model. In the case of balanced cut and SSE problems, when the partitions themselves have enough expansion within them, they can recover the planted cut upto a small error.

Louis and Venkat \cite{DBLP:conf/icalp/LouisV18} study the problem of balanced vertex expansion in a natural semi-random model and get a bi-criteria approximation algorithm for the same. They even get an exact recovery for a restricted set of parameters in their model. Their proof consisted of constructing an optimal solution to the dual of the SDP relaxation and using it to show the integrality of the optimal primal solution. In \cite{DBLP:conf/fsttcs/LouisV19}, they study the problem for a general, balanced $k-$way vertex (and edge) expansion and give efficient algorithms for the same. Their construction consists of $k$ (almost) regular expander graphs (over vertices $\set{S_i}_{i=1}^{k}$, each of size $n/k$) and then adding edges across them ensuring that the expansion of each of the $G\brac{S_i}'s$ is small. Their algorithm is based on rounding a SDP relaxation and then showing that the vertices of each $S_i$ are ``clustered'' together around the mean vector $\mu_i$ and for different sets $S_i$ and $S_j$, $\mu_i$ and $\mu_j$ are sufficiently apart. This gives a way to recover a good solution. Our approach also shows that the SDP vectors for the vertices in $S$ are ``clustered'' together. However arriving at such a conclusion requires different ideas because of the new challenges posed by the nature of the problem and assumptions on our models.

\subsection{SDP formulation}
\label{sec:sdp}
We use the following Semidefinite/Vector Programming relaxation for our problem, over the vectors $X_i$ $(i \in [n])$ and $I$.
\begin{SDP}
	\label{sdp:dks}
	\begin{align}
		\textbf{maximize}\qquad\qquad\qquad
		\label{eq:sdp1}
		\dfrac{1}{2}\sum\limits_{i,j=1}^{n}  A_{ij}\inprod{X_i, X_j}& \\
		\textbf{subject to}\qquad\qquad\qquad\qquad~
		\label{eq:sdp2}
		\sum\limits_{i=1}^{n} \inprod{X_i, X_i} &= k \\
		\label{eq:sdp3}
		\sum\limits_{j=1}^{n} \inprod{X_i, X_j} &\leq k\inprod{X_i, X_i} & \forall i \in [n]\\
		\label{eq:sdp4}
		0 \leq	\inprod{X_i, X_j} &\leq \inprod{X_i, X_i} & \forall i, j \in [n],\ (i \neq j)\\
		\label{eq:sdp5}
		\inprod{X_i, X_i} &\leq 1 & \forall i \in [n]\\
		\label{eq:sdp7}
		\inprod{X_i, I} &= \inprod{X_i, X_i} & \forall i \in [n]\\
		\label{eq:sdp8}
		\inprod{I, I} &= 1
	\end{align}
\end{SDP}
We note that these programs can be solved efficiently using standard algorithms, like ellipsoid and interior point methods. To see, why the above \prettyref{sdp:dks} is a relaxation, let $S$ be the optimal set and $v$ be any unit vector. It is easy to verify the solution set,
\[X_{i} = \begin{cases}
      v & i \in S \\
      0 & i \in V \setminus S\\
   \end{cases}
   \qquad\text{and}\qquad I = v\mper
\]
is feasible for \prettyref{sdp:dks} and gives the objective value equal to its optimal density.

%% file: overview.tex
\subsection{Proof Overview}
Our algorithms are based on rounding an SDP relaxation (\prettyref{sdp:dks}) for the \dksprob.
At a high level, we show that most of the SDP mass is concentrated on the vertices in $S$
(\prettyref{prop:exp_edge_len}, \prettyref{prop:exp_vertex_len}).
To show this, we begin by observing that
the SDP objective value is at least $kd/2$ since the integer optimal solution to the SDP
has value at least $kd/2$. Therefore, by proving an appropriate upper bound
on the SDP value from edges in $S \times (V \setminus S)$ (\prettyref{prop:upper_bound_v_s}) and the edges in $V \setminus S$ (\prettyref{prop:upper_bound_v_s_expander}, \prettyref{prop:upper_bound_cross}),
we can get a lower bound on the SDP value from the edges inside $S$. 

The edges in $S \times (V \setminus S)$ form a random bipartite graph. We can bound the contribution towards the SDP mass from this part by bounding the contribution from the ``expected graph'' (\prettyref{lem:ub_xi_xj_s_sbar}) and the contribution from the random graph minus the expected graph (\prettyref{cor:ub_b_ij_xi_xj_s_sbar_expanded}). The contribution from the latter part can be bounded using bounds on the spectra of random matrices (\prettyref{cor:ub_spectral_norm_b}). Since the expected graph is a complete weighted graph with edge weights equal to the edge probability, the contribution from this part can be bounded using the SDP constraints (\prettyref{lem:ub_xi_xj_s_sbar}).

For \dksparams~and \dksregparams, we use a result by \cite{DBLP:conf/stoc/BhaskaraCCFV10}. They construct a feasible solution to the dual of the SDP for random graphs, thereby bounding the cost of the optimal solution of the primal. Their proof only uses a bound on the spectral gap of the graph, and therefore, holds also for expander graphs. 
Therefore, this result gives us the desired bound on the SDP value on the edges inside $V \setminus S$
in these models (\prettyref{prop:upper_bound_v_s_expander}).
We also give an alternate proof of the same result using the spectral properties of the adjacency matrix of $V \setminus S$ in \prettyref{lem:svd_expander}; this approach is similar in spirit to the proof of the classical {\em expander mixing lemma}.

For \dkssparams~and \dkssregparams, we bound the SDP value on the edges inside $V \setminus S$
using a result of Charikar \cite{Charikar:2000:GAA:646688.702972}.
This work showed that for a graph $H=(V',E')$, a natural LP relaxation can be used to compute 
$\max_{W \subseteq V'} \rho(W)/\Abs{W}$.
We show that we can use our SDP solution to construct a feasible solution for 
this LP. 
Since $\rho(W)/\Abs{W} \leq \gamma d$, $\forall W\subset V \setminus S$ in this model, 
Charikar's result \cite{Charikar:2000:GAA:646688.702972} implies that the cost of 
any feasible LP solution can be bounded by $\gamma d$.
This gives us the desired bound on the SDP value on the edges inside $V \setminus S$ in these models
(\prettyref{prop:upper_bound_cross}).

These bounds establish that most of the SDP mass is on the edges inside $S$. 
Using the SDP constraints, we show that the set of vertices 
corresponding to all the ``long'' vectors will contain
a large weight of edges inside $S$ (\prettyref{cor:edges_inside_set_T}). 
Moreover, since the sum of squared lengths of the vectors is $k$
(from the SDP constraints), we can only have $\bigo{k}$ long vectors (\prettyref{lem:sizeT}).
Using standard techniques from the literature, we can prune this set to obtain a 
set of size at most $k$ and having large density \cite{Srivastav:1998:FDS:646687.702946}.
In the case when the graph induced on $S$ is $d$-regular, we show that if a set 
contains a large fraction of the edges inside $S$, then it must also have a 
large intersection with $S$ (\prettyref{lem:six} and \prettyref{lem:eight}). We present our complete procedure in \prettyref{alg:one}.

We note that while this framework for showing that the SDP mass is 
concentrated on the planted solution has been used for 
designing algorithms for semi-random instances of other problems as well, 
proving quantitative bounds is problem-specific and model-specific:
different problems and different models require different approaches.

%% file: algorithm_analysis.tex
\section{Analysis of \dksparams}
\label{sec:2}
In this section, we will analyse the \dksparams~model. Our main result is the following. 

\begin{theorem}[Formal version of \prettyref{thm:inf_main_exp}]
	\label{thm:main_exp}
	There exist universal constants $\probabilityconstant, \matrixconstant \in \mathbb{R}^{+}$ and a deterministic polynomial time algorithm, which takes an instance of \dksparams~where 
	\[\nu = 2\sqrt{3\paren{6\delta + \matrixconstant\sqrt{\dfrac{\delta n}{dk}} + \dfrac{\lambda}{d} + \dfrac{d'k}{\paren{n-k}d}}},\] 
	satisfying $\nu \in (0,1)$, and $\ffrac{\delta d}{k} \in [\probabilityconstant \ffrac{\log n}{n}, 1)$, and outputs with high probability (over the instance) a vertex set $\mathcal{Q}$ of size $k$ such that
	\[\rho(\mathcal{Q}) \geq \paren{1-\nu} \dfrac{kd}{2}\mper\]
	The above algorithm also computes a vertex set $T$ such that
	\begin{multicols}{2}
		\begin{enumerate}[(a)]
			\item $\abs{T} \leq k\paren{1+\dfrac{\nu}{5}}\mper$
			\item $\rho(T \cap S) \geq \paren{1-\dfrac{\nu}{2}} \dfrac{kd}{2}\mper$
		\end{enumerate}
	\end{multicols}
\end{theorem}

In the analysis below, without loss of generality we can ignore the adversarial action (\prettyref{step:four} of the model construction) to have taken place. Let us assume the montone adversary removes edges arbitrarily from the subgraphs $G[V \setminus S]~\&~G[S, V \setminus S]$ and the new resulting adjacency matrix is $A'$. Then for any feasible solution $\set{\set{Y_i}_{i=1}^{n}, I}$ of the SDP, we have $\sum\limits_{i \in P, j \in Q} A'_{ij}\inprod{Y_i, Y_j} \leq \sum\limits_{i \in P, j \in Q} A_{ij} \inprod{Y_i, Y_j}$ for $\forall P, Q \subseteq V$. This holds because of the non-negativity constraint \prettyref{eq:sdp4}. Thus the upper bounds on SDP contribution by vectors in $G\brac{S, V \setminus S}$ and $G\brac{V \setminus S}$  as claimed by \prettyref{prop:upper_bound_v_s} and \prettyref{prop:upper_bound_cross} respectively are intact and the rest of the proof follows exactly. Hence, without loss of generality, we can ignore this step in the analysis of our algorithm.

\subsection{Edges between $S$ and $V \setminus S$}
In this section, we show an upper bound on $\sum\limits_{i \in S, j \in V \setminus S} A_{ij} \inprod{X_i, X_j}$. 
\begin{proposition}
\label{prop:upper_bound_v_s}
W.h.p. (over the choice of the graph), we have
\[	\sum\limits_{i \in S, j \in V \setminus S} A_{ij}\inprod{X_i,X_j} \leq 3pk^2\paren{1-\E_{i \sim S}\norm{X_i}^2} 
+ \matrixconstant k\sqrt{np}\sqrt{\paren{\E_{i \sim S} \norm{X_i}^2} \paren{1-\E_{i \sim S} \norm{X_i}^2}} \mper \]
\end{proposition}
Note that 
\begin{equation}
\label{eq:aij}
\sum\limits_{i \in S, j \in V \setminus S} A_{ij} \inprod{X_i, X_j} = p\sum\limits_{i \in S, j \in V \setminus S} \inprod{X_i, X_j} + \sum\limits_{i \in S, j \in V \setminus S} (A_{ij}-p) \inprod{X_i, X_j} \mper
\end{equation}
We will bound the two terms in the R.H.S. of eqn \prettyref{eq:aij} separately.
The first term relies only on the {\em expected graph} and can be bounded using the SDP constraints. 
We use bounds on the eigenvalues of random bipartite graphs to bound the second term. 

\subsubsection*{Bounding the contribution from the {\em expected graph}}
We first prove some properties of the SDP solutions that we will use to bound this term.
The following lemma shows that if the expected value of the squared norm of the vectors corresponding to the set $S$ is ``large'', then their expected pairwise inner product is ``large'' as well. 
\begin{lemma}
	\label{lem:lb_exp_xi_xj_gen}
	Let $\set{\set{Y_i}_{i=1}^{n}, I}$ be any feasible solution of \prettyref{sdp:dks} and $T \subseteq V$ such that, $\E\limits_{i \sim T}\norm{Y_i}^2 \geq 1-\e$ where $0 \leq \e \leq 1$, then
	$\E\limits_{i,j \sim T} \inprod{Y_i, Y_j} \geq 1-4\e$.
\end{lemma}
\begin{proof}
	We first introduce vectors $Z_i$ and scalars $\alpha_i \in \R$ (for all $i \in [n]$) such that $Y_i = \alpha_i I + Z_i$ and $\inprod{I, Z_i} = 0$. 
	Using \prettyref{eq:sdp7} we get 
	\begin{align*}
		\norm{Y_i}^2 = \inprod{Y_i, I} = \inprod{\alpha_i I + Z_i, I} = \alpha_i \inprod{I,I} + \inprod{I,Z_i} = \alpha_{i}\mper
	\end{align*}
	Next,
	\begin{align}
	\nonumber
	\norm{Y_i}^2 &= \alpha_i^2\norm{I}^2 + \norm{Z_i}^2 = \norm{Y_i}^4 + \norm{Z_i}^2\\
	\label{eq:y_val}
	\implies \norm{Z_i} &= \sqrt{\norm{Y_i}^2 - \norm{Y_i}^4}.
	\end{align}
	For $i, j \in T$,
	\begin{align*}
	\inprod{Y_i, Y_j} &= \inprod{\norm{Y_i}^2I+Z_i,\norm{Y_j}^2I+Z_j}\\
	&= \norm{Y_i}^2\norm{Y_j}^2 \inprod{I, I} + \norm{Y_i}^2\inprod{I, Z_j} + \norm{Y_j}^2\inprod{I, Z_i} + \inprod{Z_i, Z_j}\\ 
	&= \norm{Y_i}^2\norm{Y_j}^2 + \inprod{Z_i, Z_j} \qquad (\because \inprod{I, Z_i} = 0)\\
	&\geq \norm{Y_i}^2\norm{Y_j}^2 - \norm{Z_i}\norm{Z_j} \qquad (\text{since the maximum angle between them can be }\pi)\\
	&= \norm{Y_i}^2\norm{Y_j}^2 - \left(\sqrt{\norm{Y_i}^2 - \norm{Y_i}^4}\right)\left(\sqrt{\norm{Y_j}^2 - \norm{Y_j}^4}\right). \qquad (\text{by eqn } \prettyref{eq:y_val})
	\end{align*}
	Summing both sides $\forall i, j \in T$ and dividing by $\abs{T}^2$,
	\[
	\sum\limits_{i, j \in T} \dfrac{\inprod{Y_i, Y_j}}{\abs{T}^2} \geq \left(\sum\limits_{i \in T} \dfrac{\norm{Y_i}^2}{\abs{T}}\right)\left(\sum\limits_{j \in T} \dfrac{\norm{Y_j}^2}{\abs{T}}\right) - \left(\sum\limits_{i \in T}\dfrac{\sqrt{\norm{Y_i}^2 - \norm{Y_i}^4}}{\abs{T}}\right)\left(\sum\limits_{j \in T} \dfrac{\sqrt{\norm{Y_j}^2 - \norm{Y_j}^4}}{\abs{T}}\right) .
	\]
	\begin{align*}
	\therefore\E_{i, j \sim T} \inprod{Y_i, Y_j} &\geq \left(\E_{i \sim T} \norm{Y_i}^2\right)^2 - \left(\E_{i \sim T} \sqrt{\norm{Y_i}^2 - \norm{Y_i}^4}\right)^2\\
	&\geq \left(\E_{i \sim T} \norm{Y_i}^2\right)^2 - \left(\E_{i \sim T} [\norm{Y_i}^2 - \norm{Y_i}^4]\right)\\ 
	& \qquad\quad \paren{\text{by Jensen's inequality, for a random variable }U \geq 0,-\left(\E\left[\sqrt{U}\right]\right)^2 \geq -\E[U]}\\
	&\geq \left(\E_{i \sim T} \norm{Y_i}^2\right)^2 - \E_{i \sim T} \norm{Y_i}^2 + \left(\E_{i \sim T} \norm{Y_i}^2\right)^2 \qquad \paren{\because \E\norm{Y_i}^4 \geq \paren{\E\norm{Y_i}^2}^2}\\
	&= 2\left(\E_{i \sim T} \norm{Y_i}^2\right)^2 - \E_{i \sim T} \norm{Y_i}^2
	\geq 2\left(1-\e\right)^2 - 1 \qquad \paren{\because \E\norm{Y_i}^2  \leq 1} \\
	& = 1 - 4 \e + 2 \e^2 \geq 1 - 4\e \mper
	\end{align*}
\end{proof}

\begin{corollary}
	\label{cor:lb_exp_xi_xj_s}
	\[ \E_{i,j \sim S}\inprod{X_i,X_j} \geq 4\E_{i \sim S} \norm{X_i}^2 - 3 \mper \]
\end{corollary}
\begin{proof}
	Using \prettyref{lem:lb_exp_xi_xj_gen} on the set $S$ and with $\e = 1-\E\limits_{i \sim S} \norm{X_i}^2$, we get the lower bound $1-4\paren{1-\E\limits_{i \sim S} \norm{X_i}^2} = 4\E\limits_{i \sim S} \norm{X_i}^2 - 3$.
\end{proof}

We are now ready to bound the first term in eqn \prettyref{eq:aij}.
\begin{lemma}
	\label{lem:ub_xi_xj_s_sbar}
	\[ \sum\limits_{i \in S, j \in V \setminus S} \inprod{X_i,X_j} \leq 3k^2\paren{1-\E_{i \sim S} \norm{X_i}^2} \mper \]
\end{lemma}
\begin{proof}
\begin{align*}
\sum\limits_{i \in S, j \in V \setminus S} \inprod{X_i,X_j} 
& = \sum\limits_{i \in S, j \in V} \inprod{X_i,X_j}  - \sum\limits_{i \in S, j \in S} \inprod{X_i,X_j} \\ 
& \leq k\sum\limits_{i \in S} \norm{X_i}^2 - \sum\limits_{i \in S, j \in S} \inprod{X_i,X_j} \qquad (\text{by eqn \prettyref{eq:sdp3}})\\ 
&= k^2\paren{\E_{i \sim S}\norm{X_i}^2} - k^2\paren{\E_{i, j \sim S} \inprod{X_i,X_j}}\\
&\leq k^2\paren{\E_{i \sim S}\norm{X_i}^2} - k^2\paren{4\E_{i \sim S} \norm{X_i}^2 - 3} \qquad (\text{by \prettyref{cor:lb_exp_xi_xj_s}})\\
&= 3k^2\paren{1-\E_{i \sim S}\norm{X_i}^2} \mper
\end{align*}
\end{proof}

\subsubsection*{Bounding the deviation from the {\em expected graph}}
We now prove the following lemmas which we will use to bound the second term in \prettyref{eq:aij}.
Let $B$ be the $n \times n$ matrix defined as follows.
\[ B_{ij} \defeq \begin{cases}  A_{ij} - p & i \in S, j \in V \setminus S \text{ or } i\in V \setminus S, j \in S \\
		0 & \textrm{otherwise} \end{cases} \mper \]
\begin{lemma}
\label{lem:ub_b_ij_xi_xj_s_sbar}
\[ \sum\limits_{i, j \in V} B_{ij}\inprod{X_i,X_j} \leq 
	2k\norm{B}\sqrt{\paren{\E_{i \sim S} \norm{X_i}^2} \paren{1-\E_{i \sim S} \norm{X_i}^2}} \mper \]
\end{lemma}
\begin{proof}
Recall that w.l.o.g., we can assume that the SDP vectors to be of dimension $n+1$. We define two matrices $Y, Z$ each of size $(n+1) \times n$. For all $i \in S$, the vector $X_i$ is placed at the $i^{th}$ column of the matrix $Y$ while the rest of the entries of $Y$ are zero. Similarly for all $j \in V \setminus S$, the vector $X_j$ is placed at the $j^{th}$ column of the matrix $Z$ and rest of the entries of $Z$ are zero. We use $Y_i$ to denote the $i^{th}$ column vector of the matrix $Y$. Similarly, $Y^{T}_{j}$ denotes the $j^{th}$ column vector of the matrix $Y^T$.
\begin{align*}
\sum\limits_{i, j \in V} B_{ij}\inprod{X_{i},X_{j}} 
	&= \sum\limits_{i, j \in V}\sum\limits_{l=1}^{n+1} B_{ij}X_i(l)X_j(l) = \sum\limits_{l=1}^{n+1}\sum\limits_{i, j \in V} B_{ij}X_i(l)X_j(l) \\
	&= 2\sum\limits_{l=1}^{n+1} \paren{Y^{T}_{l}}^T B \paren{Z^{T}_{l}}\\ &\leq 2\sum\limits_{l=1}^{n+1} \norm{Y^{T}_{l}}\norm{Z^{T}_{l}}\norm{B} \\
	\nonumber
	&\leq 2\norm{B}\sqrt{\sum\limits_{l=1}^{n+1} \norm{Y^{T}_{l}}^2}\sqrt{\sum\limits_{l=1}^{n+1} \norm{Z^{T}_{l}}^2} \qquad (\text{by Cauchy-Schwarz inequality})\\
	&= 2\norm{B}\sqrt{\sum\limits_{i \in S} \norm{X_i}^2}\sqrt{\sum\limits_{i \in V \setminus S} \norm{X_i}^2} \qquad (\text{rewriting sum of entries using columns})\\
	\nonumber
	&= 2k\norm{B}\sqrt{\paren{\E_{i \sim S} \norm{X_i}^2}\paren{1-\E_{i \sim S} \norm{X_i}^2}} \qquad (\text{by eqn \prettyref{eq:sdp2}}) .
	\end{align*}
\end{proof}

Now, we use the following folklore result to bound $\norm{B}$.
\begin{theorem}[\cite{e86c6c373bf9484a861530abef1cb4a5}, Lemma 30]
	\label{thm:ub_spectral_norm_gen}
	Let $M$ be a symmetric matrix of size $n \times n$ with zero diagonals and independent entries such that $M_{ij} = M_{ji} \sim \Bern\paren{p_{ij}}$ for all $i<j$ with $p_{ij} \in [0,1]$. Assume $p_{ij}\paren{1-p_{ij}} \leq r$ for all $i < j$ and $nr = \Omega\paren{\log n}$. Then, with high probability (over the randomness of matrix $M$),
	\[ \norm{M - \E\brac{M}} \leq {\bigO\paren{1}}\sqrt{nr} \mper\]
\end{theorem}
\begin{corollary}
	\label{cor:ub_spectral_norm_b}
	There exists universal constants $\probabilityconstant, \matrixconstant \in \mathbb{R}^{+}$ such that if $p \in \left[\dfrac{\probabilityconstant \log n}{n}, 1\right)$, then
	\[\norm{B} \leq \matrixconstant\sqrt{np}\] with high probability (over the choice of the graph).
\end{corollary}
\begin{proof}
	Let $H$ be the adjacency matrix (symmetric) of the bipartite graph $(S, V \setminus S)$, i.e., without any edges inside $S$ or $V \setminus S$. Therefore, for $i, j \in S$ or $i, j \in V \setminus S$, $H_{ij}$ is identically $0$. Fix an $i \in S$ and $j \in V \setminus S$. We know that all such $H_{ij}$'s are independent because of the assumption of random edges being added independently. By definition, $H_{ij}$ is sampled from the Bernoulli distribution with parameter $p$ or $H_{ij} \sim \Bern\paren{p}$. For the parameter range $p \in \left[\ffrac{\probabilityconstant \log n}{n}, 1\right)$, we have $p(1-p) \leq p$ and $np = \Omega(\log n)$. We now apply \prettyref{thm:ub_spectral_norm_gen} to matrix $H$ with the parameter $r = p$ to get,
	\begin{equation}
		\label{eq:norm_h}
		\norm{H - \E\brac{H}} \leq {\bigO\paren{1}}\sqrt{np} = \matrixconstant\sqrt{np} \mper
	\end{equation}
	where $\matrixconstant \in \mathbb{R}^{+}$ is the constant from \prettyref{thm:ub_spectral_norm_gen}. By definition, we have $B = H - \E\brac{H}$. Thus by eqn \prettyref{eq:norm_h},
	\[\norm{B} = \norm{H - \E\brac{H}} \leq \matrixconstant\sqrt{np}\mper\]
\end{proof}

\begin{remark}
	Note that, \prettyref{cor:ub_spectral_norm_b} holds with high probability when $p = \Omega\paren{\log n/n}$. In the rest of the paper, we work in the range of parameters where this lower bound on $p$ is satisfied. However, we do restate it when explicitly using this bound.
\end{remark}

\begin{corollary}
\label{cor:ub_b_ij_xi_xj_s_sbar_expanded}
W.h.p. (over the choice of the graph),
\[ \sum\limits_{i, j \in V} B_{ij}\inprod{X_{i},X_{j}} \leq 
	2\matrixconstant k\sqrt{np}\sqrt{\paren{\E_{i \sim S} \norm{X_i}^2}\paren{1-\E_{i \sim S} \norm{X_i}^2}}
	\mper \]
\end{corollary}
\begin{proof}
	We get the desired result by combining \prettyref{lem:ub_b_ij_xi_xj_s_sbar} and \prettyref{cor:ub_spectral_norm_b}.
\end{proof}

We are now ready to prove \prettyref{prop:upper_bound_v_s}.
\begin{proof}[Proof of \prettyref{prop:upper_bound_v_s}]
The proof follows almost immediately by combining \prettyref{lem:ub_xi_xj_s_sbar} and \prettyref{cor:ub_b_ij_xi_xj_s_sbar_expanded}
\begin{align*}
\sum\limits_{i \in S, j \in V \setminus S} A_{ij}\inprod{X_i,X_j} 
	&= p \sum\limits_{i \in S, j \in V \setminus S} \inprod{X_i,X_j} + \sum\limits_{i \in S, j \in V \setminus S} (A_{ij} - p)\inprod{X_i,X_j}\\
	&\leq p\paren{3k^2\paren{1-\E_{i \sim S}\norm{X_i}^2}} + \dfrac{1}{2}\paren{\sum\limits_{i, j \in V} B_{ij}\inprod{X_i,X_j}} \qquad \text{(by symmetry)}\\
	&\leq 3pk^2\paren{1-\E_{i \sim S}\norm{X_i}^2} + \matrixconstant k\sqrt{np}\sqrt{\paren{\E_{i \sim S} \norm{X_i}^2}\paren{1-\E_{i \sim S} \norm{X_i}^2}} \mper
	\end{align*}
\end{proof}

\subsection{Edges in $V \setminus S$}
\label{sec:2.2}

We recall, the subgraph $G\brac{V \setminus S}$ is a $(d', \lambda)-$expander in the \dksparams~ model. We show the following upper bound on the SDP mass contribution by the vectors in $V \setminus S$. 
\begin{proposition}
	\label{prop:upper_bound_v_s_expander}
	\[ \sum\limits_{i, j \in V \setminus S} A_{ij}\inprod{X_i, X_j} \leq \paren{\lambda k + \dfrac{d'k^2}{n-k}} \left(1 - \E_{i \sim S} \norm{X_i}^2\right) \mper \]
\end{proposition}

To prove the above proposition, we use the following results from the Bhaskara \etal \cite{DBLP:conf/stoc/BhaskaraCCFV10} paper.

\begin{lemma}[\cite{DBLP:conf/stoc/BhaskaraCCFV10}, Theorem 6.1]
	\label{lem:sdp-random-graph}
	For a $G(n, p)$ (\Erdos-\Renyi model) graph, the value of the SDP (\prettyref{sdp:dks}) is at most $k^2p + \bigo{k\sqrt{np}}$ with high probability when $p = \Omega\paren{\ffrac{\log n}{n}}$.
\end{lemma}

\begin{lemma}[\cite{DBLP:conf/stoc/BhaskaraCCFV10}, Theorem 6.1]
	\label{lem:sdp-expander}
	For a $(d', \lambda)$-expander graph on $n$ vertices, the value of the SDP (\prettyref{sdp:dks}) is at most $\dfrac{k^2d'}{n} + k\lambda \mper$
\end{lemma}

We note that, though the statement proved in \cite{DBLP:conf/stoc/BhaskaraCCFV10} is about random graphs (\prettyref{lem:sdp-random-graph}), their proof follows as is for an expander graph. Since, we are only applying \prettyref{lem:sdp-expander}
to the subgraph $G\brac{V \setminus S}$, we use a scaling factor of $\paren{1- \E_{i \sim S} \norm{X_i}^2}$. 
The proof of \prettyref{prop:upper_bound_v_s_expander} follows directly from the above lemma. We also provide an alternate proof of \prettyref{prop:upper_bound_v_s_expander} in \prettyref{app:a}.

\begin{remark}
	\label{rem:using_random_graph}
	If the subgraph, $G\brac{V \setminus S}$ is a random graph $(G(n-k, p))$ as considered in our discussion in \prettyref{sec:related}, we can analogously use \prettyref{lem:sdp-random-graph} to get upper bounds on $\sum_{i, j \in V \setminus S} A_{ij}\inprod{X_i,X_j}$.
\end{remark}

\subsection{Putting things together}
\label{sec:2.3}
We have shown upper bounds on the SDP mass from the edges in $S \times (V\setminus S)$ (\prettyref{prop:upper_bound_v_s})
and from the edges in $V \setminus S$ (\prettyref{prop:upper_bound_v_s_expander}).
We combine these results to show that the average value of $\inprod{X_u, X_v}$ where $\set{u, v} \in E\paren{G\brac{S}}$ is  ``large" (\prettyref{prop:exp_edge_len}). 
The SDP constraint \prettyref{eq:sdp4} implies the corresponding vertices, $u$ and $v$ have large squared norms as well. 
This immediately guides us towards a selection criteria/recovery algorithm. However we need to output a vertex set of size at most $k$, we prune this set using a greedy strategy (\prettyref{alg:one}).

\begin{lemma}
	\label{lem:exp_edge_sdp_sum}
	\[\sum\limits_{i, j \in S} A_{ij}\inprod{X_i, X_j} = \paren{kd} \E\limits_{\set{i,j} \sim E(G[S])} \inprod{X_i, X_j}\mper\]
\end{lemma}
\begin{proof}
\[\sum\limits_{i, j \in S} A_{ij}\inprod{X_i, X_j} = \sum\limits_{\set{i,j} \in E(S)} 2w\paren{\set{i,j}} \inprod{X_i, X_j}\\ 
		= \paren{kd} \E\limits_{\set{i,j} \sim E(G[S])} \inprod{X_i, X_j} \mper\]
\end{proof}

\begin{proposition}
	\label{prop:exp_edge_len}
	W.h.p. (over the choice of the graph), we have $\E\limits_{\set{i,j} \sim E(G[S])} \inprod{X_i, X_j} \geq 1 - \eta$, where
	\[ \eta = 6\delta +\matrixconstant\sqrt{\dfrac{\delta n}{dk}} + \dfrac{\lambda}{d} + \dfrac{d'k}{(n-k)d} \mper \]
\end{proposition}
\begin{proof}
 From \prettyref{prop:upper_bound_v_s} we get 
\begin{align}
\label{eq:weak_upper_bound_v_s}
	\nonumber
	\sum\limits_{i \in S, j \in V \setminus S} A_{ij}\inprod{X_i,X_j} 
	&\leq 3pk^2\paren{1-\E_{i \sim S}\norm{X_i}^2} 
	+ \matrixconstant k\sqrt{np}\sqrt{\paren{\E_{i \sim S} \norm{X_i}^2}\paren{1-\E_{i \sim S} \norm{X_i}^2}}\\
	&\leq 3pk^2 + \dfrac{\matrixconstant k\sqrt{np}}{2} \mper
	\end{align}
	The inequality above follows from the observation that $a(1-a) \leq 1/4\ \forall a \in [0,1]$.
From  \prettyref{prop:upper_bound_v_s_expander} we get
\begin{equation}
	\label{eq:weak_upper_bound_cross}
	\sum\limits_{i, j \in V \setminus S} A_{ij}\inprod{X_i,X_j} 
	\leq \paren{\lambda k + \dfrac{d'k^2}{n-k}} \left(1 - \E_{i \sim S} \norm{X_i}^2\right)
	\leq \lambda k + \dfrac{d'k^2}{n-k} \mper
	\end{equation}
Since \prettyref{sdp:dks} is a relaxation of the problem, we have $\dfrac{kd}{2} \leq \dfrac{1}{2}\sum\limits_{i, j \in V} A_{ij}\inprod{X_i, X_j}\mper$ 
In other words
	\begin{align*}
	kd &\leq \sum\limits_{i, j \in S} A_{ij}\inprod{X_i, X_j} + 2\sum\limits_{i \in S, j \in V \setminus S} A_{ij}\inprod{X_i, X_j} + \sum\limits_{i, j \in V \setminus S} A_{ij}\inprod{X_i, X_j}\\
	&\leq \sum\limits_{i, j \in S} A_{ij}\inprod{X_i, X_j} + 6pk^2 
		+ \matrixconstant k\sqrt{np} + \lambda k + \dfrac{d'k^2}{n-k}
		\qquad \text{(by eqns \prettyref{eq:weak_upper_bound_v_s} and 
		\prettyref{eq:weak_upper_bound_cross})} \mper
	\end{align*}
Therefore
	\begin{align*}
	\sum\limits_{i, j \in S} A_{ij}\inprod{X_i, X_j} &\geq kd - 6pk^2 - \matrixconstant k\sqrt{np} - \lambda k - \dfrac{d'k^2}{n-k}\\ 
	&= \left(1 - \dfrac{6pk}{d} - \dfrac{\matrixconstant\sqrt{np}}{d} - \dfrac{\lambda}{d} - \dfrac{d'k}{(n-k)d} \right)kd\\ 
	&= \left(1 - 6\delta -\matrixconstant\sqrt{\dfrac{\delta n}{dk}} - \dfrac{\lambda}{d} - \dfrac{d'k}{(n-k)d}\right)kd \qquad (\text{substituting }p = \ffrac{\delta d}{k})\\
	&= (1-\eta)kd \mper
	\end{align*}
	Dividing both sides by $kd$ and using \prettyref{lem:exp_edge_sdp_sum} completes the proof. 
\end{proof}
Now, we present the complete algorithm below.

\RestyleAlgo{boxruled}
\begin{minipage}{0.95\linewidth}
	\begin{algorithm}[H]
		\caption{Recovering a dense set $\mathcal{Q}$.}
		\label{alg:one}
		\begin{algorithmic}[1]
			\REQUIRE An Instance of \dksparams~/~\dksregparams~/~
			\dkssparams~/\\\dkssregparams~and a parameter $0 < \eta < 1$.
			\ENSURE A vertex set $\mathcal{Q}$ of size $k$.
			\STATE Solve \prettyref{sdp:dks} to get the vectors $\set{\set{X_i}_{i=1}^{n}, I}$.
			\STATE $\alpha =
			\begin{cases}  
			\ffrac{1}{\sqrt{3\eta}} & \text{For instances of type, \dksparams~or \dkssparams~}.
			\\
			\ffrac{2}{\sqrt{\eta}} & \textrm{For instances of type, \dksregparams~or \dkssregparams~}.
			\end{cases}$
			\STATE Let $T = \set{i \in V : \norm{X_i}^2 \geq 1-\alpha\eta}.$
			\STATE Initialize $\mathcal{Q} = T$.
			\IF{$\abs{\mathcal{Q}} < k$}
			\STATE Arbitrarily add remaining vertices to set $\mathcal{Q}$ to make its size $k$.
			\ELSE
			\WHILE{$\abs{\mathcal{Q}} \neq k$}
			\STATE Remove the minimum weighted vertex from the set $\mathcal{Q}$.
			\ENDWHILE
			\ENDIF
			\STATE Return $\mathcal{Q}$.
		\end{algorithmic}
	\end{algorithm}
\end{minipage}
\\~\\

Note that if $\eta = 0$, the SDP returns an integral solution and we can recover the set $S$ exactly. Therefore, w.l.o.g. we assume $\eta \neq 0, 1$.

To analyse the cost of the solution returned by \prettyref{alg:one}, we define two sets as follows.
\[ T' \defeq \set{\set{i,j} \in E : \inprod{X_i, X_j} \geq 1-\alpha\eta} 
\qquad  \text{and} \qquad T \defeq \set{i \in V : \norm{X_i}^2 \geq 1-\alpha\eta} \mcom \] 
where $1 < \alpha < \ffrac{1}{\eta}$ is a parameter to be fixed later. 

We show that a {\em large} weight of the edges inside $S$ also lies in the set $T'$. 
\begin{lemma}
\label{lem:wt_edges_t'}
W.h.p. (over the choice of the graph), 
\[ \sum\limits_{e \in T' \cap E\paren{G[S]}} w(e) \geq \dfrac{kd}{2}\left(1-\dfrac{1}{\alpha}\right) \mper \] 
\end{lemma}
\begin{proof}
	By \prettyref{prop:exp_edge_len},
	\begin{align}
	\label{eq:eta2}
	\E_{\set{i,j} \sim E(G[S])} \inprod{X_i, X_j} \geq 1 - \eta
	\implies\E_{\set{i,j} \sim E(G[S])} \brac{1-\inprod{X_i, X_j}} \leq \eta\mper
	\end{align}
	Note that by eqns \prettyref{eq:sdp4} and \prettyref{eq:sdp5}, 
	$1-\inprod{X_i, X_j} \in [0,1]\ \forall i,j \in V$.
	Therefore,
	\begin{align*}
	\ProbOp_{\set{i,j} \sim E(G[S])} \brac{\inprod{X_i, X_j} \leq 1-\alpha\eta} 
	&= \ProbOp_{\set{i,j} \sim E(G[S])} \brac{1-\inprod{X_i, X_j} \geq \alpha\eta} \\
	&\leq \dfrac{\E\limits_{\set{i,j} \sim E(G[S])} \brac{1- \inprod{X_i, X_j}}}{\alpha\eta} 
		\qquad \text{(by Markov's inequality)}\\
	&\leq \dfrac{1}{\alpha} \qquad (\text{by eqn \prettyref{eq:eta2}}).
	\end{align*}
Hence,
	\begin{align}
		\label{eq:eta3}
		\ProbOp_{\set{i,j} \sim E(G[S])} \brac{\set{i,j} \in T'} 
		= \ProbOp_{\set{i,j} \sim E(G[S])} \brac{\inprod{X_i, X_j} \geq 1-\alpha\eta}  
		\geq 1-\dfrac{1}{\alpha} \mper
	\end{align}
	The probability of the above event can be rewritten as,
	\begin{align*}
		\ProbOp_{\set{i,j} \sim E(G[S])} \brac{\set{i,j} \in T'} = \sum_{e \in E(G[S])} \one_{\set{e \in T' \cap E(G[S])}}\cdot\dfrac{w(e)}{\rho(S)} = \sum_{e \in T' \cap E(G[S])} \dfrac{w(e)}{\rho(S)} = \sum_{e \in T' \cap E(G[S])} \dfrac{w(e)}{kd/2}\mper
	\end{align*}
	By eqn \prettyref{eq:eta3},
	\begin{align*}
	\sum\limits_{e \in T' \cap E\paren{G[S]}} \dfrac{w(e)}{kd/2} \geq 1-\dfrac{1}{\alpha}
	\implies \sum\limits_{e \in T' \cap E\paren{G[S]}} w(e) \geq \dfrac{kd}{2}\left(1-\dfrac{1}{\alpha}\right)\mper
	\end{align*}
\end{proof}

The following lemma shows that the subgraph induced on $T \cap S$ contains all the edges in $T' \cap E\paren{G[S]}$.
\begin{lemma}	
\label{lem:edges_subset}
W.h.p. (over the choice of the graph),
	\[ T' \cap E\paren{G[S]} \subseteq E(G[T \cap S]) \mper \] 
\end{lemma}
\begin{proof}
	Consider any edge $e = (u,v) \in T' \cap E\paren{G[S]}$. By definition, $\inprod{X_u, X_v} \geq 1-\alpha\eta$. But by the SDP constraint \prettyref{eq:sdp4}, $\paren{\inprod{X_i, X_j} \leq \inprod{X_i, X_i}}$, we have $\norm{X_u}^2, \norm{X_v}^2 \geq 1-\alpha\eta$. Thus $u, v \in T$. Also, since $e \in E\paren{G[S]}$, so $u, v \in S$.
\end{proof}

\begin{corollary}
W.h.p. (over the choice of the graph),
	\label{cor:edges_inside_set_T}
	\[ \rho\paren{T} \geq \rho\paren{T \cap S} \geq \dfrac{kd}{2}\left(1-\dfrac{1}{\alpha}\right)\mper \]
\end{corollary}
\begin{proof}
	By \prettyref{lem:wt_edges_t'} and \prettyref{lem:edges_subset}.
\end{proof}

We have shown that the subgraph induced on $T$ has a large weight $\paren{\approx kd/2}$. In the next lemma, we show that the size of set $T$ is not too large compared to $k$.
\begin{lemma}
\label{lem:sizeT}
W.h.p. (over the choice of the graph),
	\[ \abs{T} \leq \dfrac{k}{1-\alpha\eta} \mper \] 
\end{lemma}
\begin{proof}
	Since $\sum\limits_{i \in V} \norm{X_i}^2 = k$, the number of vertices which have a squared norm greater than or equal to $1-\alpha\eta$ (precisely the set $T$) cannot exceed the bound in the lemma. 
\end{proof}

To prune the set $T$ and obtain a set of size $k$, 
we use a lemma from the work by Srivastav \etal \cite{Srivastav:1998:FDS:646687.702946}.
\begin{lemma}[\cite{Srivastav:1998:FDS:646687.702946}, Lemma 1] 
	\label{lem:pruning}
	Let $V', V'' \subseteq V$ be non-emply subsets such that $\abs{V''} \geq \abs{V'}$, then the greedy procedure which picks the lowest weighted vertex from $V''$ and removes it iteratively till we have $\abs{V'}$ vertices left ensures,
	
	\[ \rho\paren{V'} \geq \dfrac{\abs{V'}\paren{\abs{V'}-1}}{\abs{V''}\paren{\abs{V''}-1}}~\rho\paren{V''} \mper \]
\end{lemma}

We are now ready to prove the main result which gives the approximation guarantee of our algorithm. We also set the value of parameter $\alpha$ which maximizes the density of the output graph.

\begin{proof}[Proof of \prettyref{thm:main_exp}]
We run \prettyref{alg:one} on \dksparams~with $\eta$ as given in \prettyref{prop:exp_edge_len}. From \prettyref{lem:pruning}, we have a handle on the density of the new set $(\mathcal{Q})$ \emph{after pruning} $T$ to a set of size $k$. The algorithm performs this exactly in the steps 5 to 11.
Let $\ALG$ denote the density of this new set (output of \prettyref{alg:one}). We have,
	\begin{align*}
		\ALG &\geq \paren{\dfrac{k(k-1)}{\abs{T}(\abs{T}-1)}} \left(1-\dfrac{1}{\alpha}\right) \dfrac{kd}{2}
		& (\text{by \prettyref{cor:edges_inside_set_T} and \prettyref{lem:pruning}})\\
	&\geq \paren{\dfrac{(k-1)(1-\alpha\eta)^2}{k-1+\alpha\eta}}\paren{1-\dfrac{1}{\alpha}} \dfrac{kd}{2} & (\text{by \prettyref{lem:sizeT}})\\
	&= \paren{\dfrac{(1-\alpha\eta)^2}{1+\ffrac{\alpha\eta}{(k-1)}}}\paren{1-\dfrac{1}{\alpha}} \dfrac{kd}{2} & (\text{dividing by }k-1)\\ 
	&\geq \paren{\dfrac{(1-\alpha\eta)^2}{1+\alpha\eta}}\paren{1-\dfrac{1}{\alpha}} \dfrac{kd}{2} & (\text{w.l.o.g., }k \geq 2)\\
	&\geq \paren{1-2\alpha\eta}\paren{1-\alpha\eta}\paren{1-\dfrac{1}{\alpha}} \dfrac{kd}{2} & \paren{\because (1-x)^2 \geq 1-2x \text{ and } \dfrac{1}{1+x} \geq 1-x,~\forall x \in \mathbb{R}_{\geq 0}}\\
	&\geq \paren{1-3\alpha\eta-\dfrac{1}{\alpha}} \dfrac{kd}{2} & (\text{rearranging and bounding the positive terms by 0})\\
	&= \paren{1-2\sqrt{3\eta}} \dfrac{kd}{2} & \paren{\text{we fix }\alpha = \ffrac{1}{\sqrt{3\eta}}}.
\end{align*}
Letting $\nu \defeq 2\sqrt{3\eta}$, we get that $\ALG \geq \paren{1 - \tau} \ffrac{kd}{2}$
where
	\[ \nu = 2\sqrt{3\paren{6\delta +\matrixconstant\sqrt{\dfrac{\delta n}{dk}} + \dfrac{\lambda}{d} + \dfrac{d'k}{(n-k)d}}} \qquad \paren{\text{using the value of } \eta \text{ from \prettyref{prop:exp_edge_len}}}.
	\]
	From \prettyref{lem:sizeT}, 
	\[
		\abs{T} \leq \dfrac{k}{1-\alpha\eta} = \dfrac{k}{1-\paren{\ffrac{\nu}{6}}} \leq k\paren{1+\dfrac{\nu}{5}} \mper
	\]
	And from \prettyref{cor:edges_inside_set_T},
	\[
		\rho\paren{T \cap S} \geq \dfrac{kd}{2}\left(1-\dfrac{1}{\alpha}\right)	= \dfrac{kd}{2}\left(1-\dfrac{\nu}{2}\right)\mper
	\]
\end{proof}
Note that for the parameter range $0 < 2\sqrt{3\eta} < 1 \iff 0 < \nu < 1$, the value of $\alpha~(= \ffrac{1}{\sqrt{3\eta}})$ fixed by the algorithm lies in the interval $(1, \ffrac{1}{\eta})$ as required.

\begin{remark}[on \prettyref{thm:inf_main_exp}]
	\label{rem:simplified_apx_4}
	In the restricted parameter case, we simplify the arguments in our informal theorem statements, i.e. the case when the average degree of vertices in $S$ and $V \setminus S$ is close, we have $\delta = \Theta\paren{\dfrac{kd'}{nd}}$. Assuming $\nu = 2\sqrt{3\eta}$, we rewrite $\dfrac{\delta n}{dk}$ as $\dfrac{d'}{d^2}$ from the above value of $\delta$ and the term $\dfrac{\paren{d' - \lambda}k}{\paren{n-k}d}$ is at most a constant for ``large'' n. So, the new value of $\tau$ is  $\Theta\paren{\sqrt{\delta+\dfrac{\lambda+\sqrt{d'}}{d}}}$. A similar argument gives the new value of $\nu'$ in \prettyref{thm:inf_main_exp_reg}.
\end{remark}

%% file: other_models.tex
\section{Other Models}
\label{sec:3}

In this section, we analyse the remaining three models. As described in the previous section, the monotone adversary step during the model generation can be ignored w.l.o.g. for the analysis.

\subsection{Analysis of \dksregparams}
\label{sec:3.1}
Recall that \dksregparams~is same as \dksparams~except that the subgraph $G\brac{S}$ is required to be an arbitrary $d-$regular graph.
We prove the following theorem. 

\begin{theorem}[Formal version of \prettyref{thm:inf_main_exp_reg}]
	\label{thm:main_exp_reg}
	There exist universal constants $\probabilityconstant, \matrixconstant \in \mathbb{R}^{+}$ and a deterministic polynomial time algorithm, which takes an instance of \dksregparams\\where
	\[\nu' = \dfrac{5}{\sqrt{1+\dfrac{dk}{4\matrixconstant^2\delta}\left(1- \dfrac{\lambda}{d} - \dfrac{d'k}{\paren{n-k}d} - 6\delta\right)^2}},\]
	satisfying $\nu' \in (0,1)$, and $\ffrac{\delta d}{k} \in [\probabilityconstant \ffrac{\log n}{n}, 1)$, and outputs with high probability (over the instance) a vertex set $\mathcal{Q}$ of size $k$ such that
	\begin{multicols}{2}
		\begin{enumerate}[(a)]
			\item $\rho(\mathcal{Q}) \geq \paren{1-\nu'} \dfrac{kd}{2} \mper$
			\item $\abs{\mathcal{Q} \cap S} \geq \paren{1-\dfrac{\nu'}{6}} k \mper$
		\end{enumerate}
	\end{multicols}
\end{theorem}

Since in the model \dksregparams, the only difference is with respect to $G\brac{S}$ from \dksparams, the upper bounds on the SDP mass from the remainder graph (\prettyref{prop:upper_bound_v_s} and \prettyref{prop:upper_bound_v_s_expander}) hold here as well. The analysis is quite similar to the proof of \prettyref{thm:main_exp}.
\begin{lemma}
	\label{lem:four}
	\[ \sum\limits_{i, j \in S} A_{ij}\inprod{X_i, X_j} \leq kd\paren{\E_{i \sim S} \norm{X_i}^2}\mper \]
\end{lemma}
\begin{proof}
\begin{align*}
\paren{\frac{1}{kd}}\sum\limits_{i, j \in S} A_{ij}\inprod{X_i, X_j} 
& \leq  \paren{\frac{1}{kd}}\sum\limits_{i, j \in S} A_{ij}
 \left(\dfrac{\norm{X_i}^2 + \norm{X_j}^2}{2}\right) \qquad \paren{\text{by expanding, } \norm{X_i - X_j}^2 \geq 0} \\
&= \left(\dfrac{1}{kd}\right) \sum\limits_{i \in S}  \paren{\sum_{j\in S}A_{ij} } \norm{X_i}^2 
 = \E_{i \sim S} \norm{X_i}^2 \mper 
\end{align*}
Here the last equality follows from the $d$-regularity of the graph induced on $S$:
$\sum_{j\in S}A_{ij}  = d$ for each $i \in S$.
\end{proof}

Next we combine, the upper bounds from \prettyref{prop:upper_bound_v_s} and \prettyref{prop:upper_bound_v_s_expander} (in terms of $\E_{i \sim S} \norm{X_i}^2$) with the above lemma (\prettyref{lem:four}) to show a result similar to \prettyref{prop:exp_edge_len}, its proof is along the same lines.
\begin{proposition}
	\label{prop:exp_vertex_len}
	W.h.p. (over the choice of the graph), we have $ \E\limits_{i \sim S} \norm{X_i}^2 \geq 1 - \eta' $, where
	\[ \eta' = \dfrac{1}{1+\dfrac{dk}{4\matrixconstant^2 \delta n}\left(1- \dfrac{\lambda}{d}-\dfrac{d'k}{\paren{n-k}d} - 6\delta\right)^2}\mper \]
\end{proposition}
\begin{proof}
	Since \prettyref{sdp:dks} is a relaxation, we have 
	$\dfrac{kd}{2} \leq \dfrac{1}{2}\sum\limits_{i, j \in V} A_{ij}\inprod{X_i, X_j}$. Thus
	\begin{align*}
		kd &\leq \sum\limits_{i, j \in S} A_{ij}\inprod{X_i, X_j} + 2\sum\limits_{i \in S, j \in V \setminus S} A_{ij}\inprod{X_i, X_j} + \sum\limits_{i, j \in V \setminus S} A_{ij}\inprod{X_i, X_j}\\
		&\leq \paren{kd}\paren{\E_{i \sim S} \norm{X_i}^2} + 2\paren{3pk^2\paren{1-\E_{i \sim S}\norm{X_i}^2} + \matrixconstant k\sqrt{np}\sqrt{\paren{\E_{i \sim S} \norm{X_i}^2}\paren{1-\E_{i \sim S} \norm{X_i}^2}}}\\
		&\quad+\paren{\lambda k + \dfrac{d'k^2}{n-k}} \left(1 - \E_{i \sim S} \norm{X_i}^2\right)
		\\&\qquad\qquad\qquad\qquad\qquad\qquad\qquad\qquad (\text{by \prettyref{prop:upper_bound_v_s}, \prettyref{prop:upper_bound_v_s_expander},
		and \prettyref{lem:four}}).
	\end{align*}
Therefore
\begin{align*}
	kd\paren{1-\E_{i \sim S} \norm{X_i}^2} &\leq 6pk^2\paren{1-\E_{i \sim S} \norm{X_i}^2} + 2\matrixconstant k\sqrt{np}\sqrt{\paren{\E_{i \sim S} \norm{X_i}^2}\paren{1-\E_{i \sim S} \norm{X_i}^2}} 
	\\&\qquad\qquad+\paren{\lambda k + \dfrac{d'k^2}{n-k}} \left(1 - \E_{i \sim S} \norm{X_i}^2\right) \mper
\end{align*}
Dividing both sides by $kd\paren{1-\E_{i \sim S} \norm{X_i}^2}$, we get
\begin{align*}
	1 \leq \dfrac{6pk}{d} + \dfrac{2\matrixconstant\sqrt{np}}{d}\sqrt{\dfrac{\E_{i \sim S} \norm{X_i}^2}{1-\E_{i \sim S} \norm{X_i}^2}} + \dfrac{\lambda}{d}+\dfrac{d'k}{\paren{n-k}d}\mper
\end{align*}
Rearranging the terms, we get
\begin{align*}
	\sqrt{\dfrac{\E_{i \sim S} \norm{X_i}^2}{1-\E_{i \sim S} \norm{X_i}^2}} &\geq \dfrac{d}{2\matrixconstant\sqrt{np}}\left(1- \dfrac{\lambda}{d}-\dfrac{d'k}{\paren{n-k}d} - \dfrac{6pk}{d}\right)\\
	\implies \E_{i \sim S} \norm{X_i}^2 &\geq 1 - \dfrac{1}{1+\paren{\dfrac{d}{2\matrixconstant\sqrt{np}}\left(1- \dfrac{\lambda}{d} - \dfrac{d'k}{\paren{n-k}d} - \dfrac{6pk}{d}\right)}^2}\\
	&= 1-\eta' \mcom
\end{align*}
	where by substituting $p=\ffrac{\delta d}{k}$ in the second last step completes the proof.
\end{proof}
Consider the vertex set \[T \defeq \set{i \in V : \norm{X_i}^2 \geq 1 - \alpha \eta'} \] where $2 < \alpha < \ffrac{1}{\eta'}$ is a parameter to be chosen later. 

In the remainder of this section we show that the subgraph induced on $T$ has a ``large" fraction of the weight of edges in $G\brac{S}$. In the next three lemmas, we show that we can recover most of the vertices of $S$ since $\abs{T \cap S}$ is large. Combining this with the regularity condition, we estimate $\rho\paren{T \cap S}$.
\begin{lemma}
\label{lem:five}
W.h.p. (over the choice of the graph), 
for $l < 1$,
	\[ \ProbOp_{i \sim S}\brac{\norm{X_{i}}^2 \geq l} \geq  1 - \dfrac{\eta'}{1-l}\mper\]
\end{lemma}
\begin{proof}
	By \prettyref{prop:exp_vertex_len},
	\begin{equation}
	\label{eq:eta'}
	\E_{i \sim S}\norm{X_{i}}^2 \geq 1-\eta'
	\implies \E_{i \sim S}\brac{1-\norm{X_{i}}^2} \leq \eta'.
	\end{equation}
	Note that by eqns \prettyref{eq:sdp4} and \prettyref{eq:sdp5}, 
	$1-\norm{X_{i}}^2 \in [0,1]\ \forall i \in V$. Therefore,
	\begin{align*}
		\ProbOp_{i \sim S}\brac{\norm{X_{i}}^2 \leq l} &= \ProbOp_{i \sim S}\brac{1-\norm{X_{i}}^2 \geq 1-l} \leq \dfrac{\E\limits_{i \sim S}\brac{1- \norm{X_{i}}^2}}{1-l} \leq \dfrac{\eta'}{1-l},
\end{align*}
where we used the Markov's inequality in the second last step and eqn \prettyref{eq:eta'} in the last step.
\end{proof}

\begin{lemma}
\label{lem:six}
W.h.p. (over the choice of the graph), 
	\[ \left(1 - \dfrac{1}{\alpha}\right)k \leq \abs{T \cap S} \leq k \mper\]
\end{lemma}
\begin{proof}
	Note, $\abs{T \cap S} \leq \abs{S} = k$.
	To prove the other inequality, we invoke \prettyref{lem:five} with $l = 1-\alpha\eta'$ (in the second step) to get,
	\begin{align}
		\label{eq:six_1}
		\ProbOp_{i \sim S} \brac{i \in T} = \ProbOp_{i \sim S}\brac{\norm{X_{i}}^2 \geq 1-\alpha \eta'} \geq 1-\dfrac{\eta'}{\alpha\eta'} = 1 - \dfrac{1}{\alpha}\mper
	\end{align}
	The probability of the above event can be rewritten as,
	\begin{align*}
	\ProbOp_{i \sim S} \brac{i \in T} = \sum_{i \in S} \one_{\set{i \in T \cap S}}\cdot\dfrac{1}{\abs{S}} = \dfrac{\abs{T \cap S}}{\abs{S}} = \dfrac{\abs{T \cap S}}{k}\mper
	\end{align*}
	By eqn \prettyref{eq:six_1},
	\begin{align*}
	\dfrac{\abs{T \cap S}}{k} \geq 1-\dfrac{1}{\alpha}
	\implies \abs{T \cap S} \geq k\paren{1-\dfrac{1}{\alpha}}\mper
	\end{align*}
\end{proof}

\begin{lemma}
	\label{lem:eight}
W.h.p. (over the choice of the graph), 
	\[ \rho\paren{T \cap S} \geq \dfrac{kd}{2}\left(1-\dfrac{2}{\alpha}\right) \mper \]
\end{lemma}
\begin{proof}
	From \prettyref{lem:six}, we get $\abs{\overline{T} \cap S} \leq \dfrac{k}{\alpha}$. We obtain a lower bound on $\rho\paren{T \cap S}$ by subtracting the weight of all edges completely inside $G\brac{\overline{T} \cap S}$ and those across the subgraphs $G\brac{T \cap S}$ and $G\brac{\overline{T} \cap S}$. We can upper bound this by $\dfrac{k}{\alpha} \cdot d$, since $S$ is $d$-regular.
	\begin{align*}
		\therefore \rho\paren{T \cap S} \geq \dfrac{kd}{2}-\dfrac{kd}{\alpha} = \dfrac{kd}{2}\left(1-\dfrac{2}{\alpha}\right)\mper
	\end{align*} 
\end{proof}

We note that the result of the above lemma is similar to what we obtained in the \dksparams\\model (\prettyref{cor:edges_inside_set_T}). Using the next two lemmas, we give an upper bound on $\abs{T}$.

\begin{lemma}
W.h.p. (over the choice of the graph),
	\label{lem:seven}
	\[ \abs{T \cap (V \setminus S)} \leq \dfrac{\eta' k}{1-\alpha \eta'}\mper\]
\end{lemma}
\begin{proof}
	We know from \prettyref{eq:sdp2} that 
	$\sum\limits_{i \in S} \norm{X_i}^2 + \sum\limits_{i \in V \setminus S} \norm{X_i}^2 = k$. 
	Using \prettyref{prop:exp_vertex_len}, we get that 
	\[ \sum\limits_{i \in V \setminus S} \norm{X_i}^2 = k - \sum\limits_{i \in S} \norm{X_i}^2 \leq k - (1-\eta')k = \eta' k.
	\]
	Therefore, the number of vertices of $V \setminus S$ which have squared norm greater than $1-\alpha\eta'$ cannot exceed the bound in the lemma.
\end{proof}

\begin{corollary}
	\label{cor:two}
	By \prettyref{lem:six} and \prettyref{lem:seven},
	\[ \abs{T} = \abs{T \cap S} + \abs{T \cap (V \setminus S)} \leq k\left(1+\dfrac{\eta'}{1-\alpha \eta'}\right) \mper\]
\end{corollary}

We are now ready to present the proof of the main result.

\begin{proof}[Proof of \prettyref{thm:main_exp_reg}]
	This proof is based exactly on the proof of \prettyref{thm:main_exp}. We will run the \prettyref{alg:one} on \dksregparams~with $\eta'$ as given in \prettyref{prop:exp_vertex_len}. Thus, after pruning the set $T$, we find its cost as follows.
	\begin{align*}
		\ALG &\geq \paren{\dfrac{k(k-1)}{\abs{T}(\abs{T}-1)}} \dfrac{kd}{2}\left(1-\dfrac{2}{\alpha}\right) & (\text{by \prettyref{lem:eight} and \prettyref{lem:pruning}})\\
		&\geq \paren{\dfrac{(k-1)(1-\alpha\eta')^2}{(1+\eta'-\alpha\eta')(k-1+k\eta'-\alpha\eta'(k-1))}}\paren{1-\dfrac{2}{\alpha}} \dfrac{kd}{2} & (\text{by \prettyref{cor:two}})\\
		&=\paren{\dfrac{(1-\alpha\eta')^2}{(1+\eta'-\alpha\eta')(1+(\ffrac{k}{k-1})\eta'-\alpha\eta')}}\paren{1-\dfrac{2}{\alpha}} \dfrac{kd}{2} & (\text{dividing by }k-1)\\
		&\geq\paren{\dfrac{1-\alpha\eta'}{1+2\eta'-\alpha\eta'}}^2\paren{1-\dfrac{2}{\alpha}} \dfrac{kd}{2} & \paren{\text{w.l.o.g., }k \geq 2 \iff \dfrac{k}{k-1} \leq 2}\\
		&\geq\paren{1-\alpha\eta'}^2\paren{1-\dfrac{2}{\alpha}} \dfrac{kd}{2} & \paren{\text{we seek an } \alpha > 2}\\
		&\geq\paren{1-2\alpha\eta'}\paren{1-\dfrac{2}{\alpha}} \dfrac{kd}{2} & (\because (1-x)^2 \geq 1-2x,~\forall x \in \mathbb{R})\\
		&=\paren{1-4\sqrt{\eta'}}\paren{1-\sqrt{\eta'}} \dfrac{kd}{2} & \paren{\text{we fix }\alpha \defeq \ffrac{2}{\sqrt{\eta'}}}\\
		&\geq\paren{1-5\sqrt{\eta'}} \dfrac{kd}{2}\mper
	\end{align*}
	Letting $\nu' \defeq 5\sqrt{\eta'}$ and using the value of $\eta'$ from \prettyref{prop:exp_vertex_len} completes the proof of the part (a).
	
	Note that for the parameter range $0 < 5\sqrt{\eta'} < 1 \iff 0 < \nu' < 1$, the value of $\alpha~(= \ffrac{2}{\sqrt{\eta'}})$ fixed by the algorithm lies in the interval $(2, \ffrac{1}{\eta'})$ as required.

	To show the part (b) of the theorem, let $\mathcal{Q}$ be the output of our algorithm $(\abs{\mathcal{Q}} = k)$ after pruning the set $T$. Note that the maximum number of vertices we remove are $k\left(\dfrac{\eta'}{1-\alpha \eta'}\right)$ from \prettyref{cor:two}. Considering the worst case, where all the removed vertices are in $\abs{T \cap S}$. We get,
\begin{align*}
	\abs{\mathcal{Q} \cap S} &\geq \abs{T \cap S} - k\left(\dfrac{\eta'}{1-\alpha \eta'}\right) \\
	&\geq \paren{1-\dfrac{1}{\alpha} - \dfrac{\eta'}{1-\alpha\eta'}}k & \paren{\text{from \prettyref{lem:six}}}\\
	&= \paren{1-\dfrac{\sqrt{\eta'}}{2} - \dfrac{\eta'}{1-2\sqrt{\eta'}}}k & \paren{\text{from the value of chosen } \alpha}\\
	&= \paren{1-\dfrac{\sqrt{\eta'}}{2(1-2\sqrt{\eta'})}}k\\
	&= \paren{1-\dfrac{\nu'}{2(5-2\nu')}}k & \paren{\text{from the value of chosen } \nu'}\\
	&\geq \paren{1-\dfrac{\nu'}{6}}k & \paren{\text{for } 0 < \nu' < 1}.
\end{align*}
\end{proof}

\subsection{Analysis of \dkssparams~and \dkssregparams}
\label{sec:3.2}

We restate the assumption on $G\brac{V \setminus S}$ in the two models for clarity, for any non-empty subset $W \subseteq V \setminus S$, $\dfrac{\rho(W)}{|W|} \leq \gamma d$. 
Given a graph $H = (V',E',w')$, consider the following LP relaxation for the problem of computing $\max\limits_{V'' \subseteq V'} \rho(V'')/\Abs{V''}$.
\begin{LP}
	\label{lp:lp}
	\begin{align}
		\textbf{maximize}\qquad\qquad\qquad\qquad\quad
		\label{eq:lp0}
		\sum\limits_{\set{i,j} \in E'} &w'_{ij}x_{ij} \\
		\textbf{subject to}\qquad\qquad\qquad\qquad\qquad\qquad
		\label{eq:lp1}
		&x_{i j} \leq y_{i} & \forall \set{i, j} \in E' \\
		\label{eq:lp2}
		&x_{i j} \leq y_{j} & \forall \set{i, j} \in E' \\
		\label{eq:lp3}
		&\sum_{i \in V'} y_{i} \leq 1 \\
		\label{eq:lp4}
		&x_{i j} \geq 0 & \forall \set{i, j} \in E' \\
		\label{eq:lp5}
		&y_{i}  \geq 0 & \forall i \in V'
	\end{align}
\end{LP}

Charikar \cite{Charikar:2000:GAA:646688.702972} proved the following result.
\begin{lemma}[\cite{Charikar:2000:GAA:646688.702972}, Theorem 1]
	\label{lem:charikar_lp}
	For a given graph $H = (V', E', w')$,
	\[\max\limits_{V'' \subseteq V'} \dfrac{\rho(V'')}{|V''|} = \OPT\paren{LP}\]
	where $\OPT\paren{LP}$ denotes the optimal value of the linear program \prettyref{lp:lp}.
\end{lemma}

Note that the constraints of both \prettyref{lp:lp} and \prettyref{sdp:dks} closely resemble each other and hence after an appropriate scaling, we can construct a feasible solution to \prettyref{lp:lp} using our SDP solution. 
\begin{lemma}
	\label{lem:lpfeasibility}
	For $G\brac{V \setminus S}$,
	\[ 
	x_{ij} := \dfrac{\inprod{X_i, X_j}}{\sum\limits_{i \in V \setminus S} \norm{X_i}^2}
	\text{ for } \set{i, j} \in E\paren{G\brac{V \setminus S}} \qquad \text{and} \qquad
	y_{i} := \dfrac{\norm{X_i}^2}{\sum\limits_{i \in V \setminus S} \norm{X_i}^2}
	\text{ for } i \in V \setminus S 
	\]
	is a feasible solution for \prettyref{lp:lp}, where $\set{\set{X_i}_{i=1}^{n}, I}$ is a feasible solution of the \prettyref{sdp:dks}.
\end{lemma}
\begin{proof}
	\begin{enumerate}
		\item  The constraints \prettyref{eq:lp1} and \prettyref{eq:lp2} and the non-negativity constraints \prettyref{eq:lp4} and \prettyref{eq:lp5} hold by the SDP constraint \prettyref{eq:sdp4}.
		\item For constraint \prettyref{eq:lp3}, 
		\\$\sum\limits_{i \in V \setminus S} y_i = \sum\limits_{i \in V \setminus S} \left(\dfrac{\norm{X_i}^2}{\sum\limits_{i \in V \setminus S} \norm{X_i}^2} \right)= \dfrac{\sum\limits_{i \in V \setminus S} \norm{X_i}^2}{\sum\limits_{i \in V \setminus S} \norm{X_i}^2} = 1$
	\end{enumerate}
\end{proof}
\begin{proposition}
	\label{prop:upper_bound_cross}
	\[ \sum\limits_{i, j \in V \setminus S} A_{ij}\inprod{X_i,X_j} \leq \left(2\gamma dk\right) \left(1 - \E_{i \sim S} \norm{X_i}^2\right) \mper\]
\end{proposition}
\begin{proof}
	By \prettyref{lem:lpfeasibility}, we know that any feasible solution to \prettyref{lp:lp} 
	satisfies
	\[ \sum\limits_{\set{i,j} \in E'} w'_{ij}x_{ij} \leq \max_{V'' \subseteq V'} \frac{\rho(V'')}{\Abs{V''}} \mper \]
	Therefore,
	\begin{align*}
	\sum\limits_{i, j \in V \setminus S} A_{ij}\inprod{X_i, X_j} &= 2\left(\sum\limits_{\set{i, j} \in E\paren{G\brac{V \setminus S}}} \dfrac{A_{ij}\inprod{X_i, X_j}}{\sum\limits_{i \in V \setminus S} \norm{X_i}^2}\right) \sum\limits_{i \in V \setminus S} \norm{X_i}^2 \\
	& \leq 2\left(\max_{W \subseteq V \setminus S} \frac{\rho(W)}{\abs{W}}\right) \sum\limits_{i \in V \setminus S} \norm{X_i}^2 
	\leq 2\left(\gamma d\right) \sum\limits_{i \in V \setminus S} \norm{X_i}^2 \\
	& = \left(2\gamma dk\right) \left(1 - \E_{i \sim S} \norm{X_i}^2\right) \qquad 
		(\text{using eqn \prettyref{eq:sdp2}}).
	\end{align*}
\end{proof}

Our main result for \dkssparams~is the following. 

\begin{theorem}[Formal version of \prettyref{thm:inf_main_gamma}]
	\label{thm:main_gamma}
	There exist universal constants $\probabilityconstant, \matrixconstant \in \mathbb{R}^{+}$ and a deterministic polynomial time algorithm, which takes an instance of \dkssparams~where
	\[  
	\tau = 2\sqrt{3\paren{6\delta + \matrixconstant\sqrt{\dfrac{\delta n}{dk}} + 2\gamma}},\] 
	satisfying $\tau \in (0,1)$, and $\ffrac{\delta d}{k} \in [\probabilityconstant \ffrac{\log n}{n}, 1)$, and outputs with high probability (over the instance) a vertex set $\mathcal{Q}$ of size $k$ such that \[ \rho(\mathcal{Q}) \geq \paren{1-\tau} \dfrac{kd}{2}\mper\]
	The above algorithm also computes a vertex set $T$ such that
	\begin{multicols}{2}
		\begin{enumerate}[(a)]
			\item $\abs{T} \leq k\paren{1+\dfrac{\tau}{5}}\mper$
			\item $\rho(T \cap S) \geq \paren{1-\dfrac{\tau}{2}} \dfrac{kd}{2}\mper$
		\end{enumerate}
	\end{multicols}
\end{theorem}

\begin{proof}
We can prove an analog of \prettyref{prop:exp_edge_len} for \dkssparams~by using the 
upper bounds on $\sum_{i, j \in V \setminus S} A_{ij}\inprod{X_i,X_j}$ 
from \prettyref{prop:upper_bound_cross} instead of \prettyref{prop:upper_bound_v_s_expander} in the proof
of \prettyref{prop:exp_edge_len}. We obtain 
\[\eta = 6\delta + \matrixconstant\sqrt{\dfrac{\delta n}{dk}} + 2\gamma \mper \]
The rest of proof is identical to the proof of \prettyref{thm:main_exp}. We set $\tau \defeq 2\sqrt{3\eta}$.
\end{proof}

\begin{remark}[on \prettyref{thm:inf_main_gamma}]
	\label{rem:simplified_apx_3}
	For the case when the average degree of vertices in $S$ and $V \setminus S$ is close (see \prettyref{rem:simplified_apx_1}), we have $\delta = \Theta\paren{\dfrac{k}{n}}$. Assuming $\tau = 2\sqrt{3\eta}$, we rewrite $\dfrac{\delta n}{dk}$ as $\Theta\paren{\dfrac{1}{d}}$ from the above value of $\delta$. So, the new value of $\tau$ is $\Theta\paren{\sqrt{\delta+\gamma+\dfrac{1}{\sqrt{d}}}}$. A similar argument gives the new value of $\tau'$ in \prettyref{thm:inf_main_gamma_reg}.
\end{remark}

\begin{theorem}[Formal version of \prettyref{thm:inf_main_gamma_reg}]
	\label{thm:main_gamma_reg}
	There exist universal constants $\probabilityconstant, \matrixconstant \in \mathbb{R}^{+}$ and a deterministic polynomial time algorithm, which takes an instance of \dkssregparams~where 
	\[\tau' = \dfrac{5}{\sqrt{1+\dfrac{dk}{4\matrixconstant^2 \delta n}\left(1- 2\gamma - 6\delta\right)^2}}, \] 
	satisfying $\tau' \in (0,1)$, and $\ffrac{\delta d}{k} \in [\probabilityconstant \ffrac{\log n}{n}, 1)$, and outputs with high probability (over the instance) a vertex set $\mathcal{Q}$ of size $k$ such that
	\begin{multicols}{2}
		\begin{enumerate}[(a)]
			\item $\rho(\mathcal{Q}) \geq \paren{1-\tau'} \dfrac{kd}{2}\mper$
			\item $\abs{\mathcal{Q} \cap S} \geq \paren{1-\dfrac{\tau'}{6}} k\mper$
		\end{enumerate}
	\end{multicols}
\end{theorem}

\begin{proof}
We can prove an analog of \prettyref{prop:exp_vertex_len} for \dkssregparams~by using the upper bounds on $\sum_{i, j \in V \setminus S} A_{ij}\inprod{X_i,X_j}$ 
from \prettyref{prop:upper_bound_cross} instead of \prettyref{prop:upper_bound_v_s_expander} in the proof
of \prettyref{prop:exp_vertex_len}. We obtain  
\[\eta' = \dfrac{1}{1+\dfrac{dk}{4\matrixconstant^2 \delta n}\left(1- 2\gamma - 6\delta\right)^2} \mper \]
The rest of proof is identical to the proof of \prettyref{thm:main_exp_reg}. 
We set $\tau' \defeq 5\sqrt{\eta'}$.
\end{proof}

%% file: appendix_2.tex
\section{Comparing $\tau \text{ and } \tau'$}
\label{app:b}

Consider two graphs, $G_1 \sim$ \dkssparams~and $G_2 \sim$ \dkssregparams~(with the same input parameters). We will show that for most natural regime of parameters $\tau = \Omega(\tau')$ and hence our algorithm gives a better guarantee for $G_2$.

\begin{lemma}
	\label{lem:comp_1}
	Let $\mathcal{D} = \set{(x, y) \in \mathbb{R}^2 : x, y > 0\text{ and } x+y < 1}$. Then,
	\[\sqrt{x+y} > \dfrac{y}{1-x} \qquad\forall (x,y) \in \mathcal{D} \mper\]
\end{lemma}
\begin{proof}
	Consider the expression,
	\begin{align*}
		x+y - \dfrac{y^2}{(1-x)^2} &= \dfrac{(x+y)(1-x)^2 - y^2}{(1-x)^2} = \dfrac{(1-(x+y))(y+x(1-x))}{(1-x)^2} > 0 \qquad\forall (x,y) \in \mathcal{D}. \\
		\therefore x+y &> \dfrac{y^2}{(1-x)^2} \implies \sqrt{x+y} > \dfrac{y}{1-x}.
	\end{align*}
\end{proof}
Recall that (upto constant factors), \[\tau = \Theta\paren{\sqrt{\delta+\gamma+\sqrt{\dfrac{\delta n}{dk}}
}} \text{ and } \tau' = \dfrac{1}{\sqrt{1+\dfrac{dk}{\delta n}\paren{1-\gamma-\delta}^2}}\mper\] Using \prettyref{lem:comp_1} with $x = \delta + \gamma$ and $y = \sqrt{\dfrac{\delta n}{dk}}$. (The condition $x+y < 1$ ensures that $\tau, \tau' < 1$), we get,
\begin{align*}
	\sqrt{\delta + \gamma + \sqrt{\dfrac{\delta n}{dk}}} > \dfrac{\sqrt{\delta n}}{\sqrt{dk}(1-\delta - \gamma)} > \dfrac{1}{\sqrt{1+\dfrac{dk}{\delta n}\paren{1-\gamma-\delta}^2}} \implies \tau = \Omega(\tau').
\end{align*}
Note that in the above calculation, we have ignored the constants in these expressions, but this is not a critical issue. A similar calculation shows that $\nu = \Omega(\nu')$ in the \dksparams~and \dksregparams~models respectively.

%% file: appendix.tex
\section{Alternate proof of \prettyref{prop:upper_bound_v_s_expander}}
\label{app:a}
For the \dksparams~and \dksregparams~models, we show the following upper bound on the SDP mass contribution by the vectors in $V \setminus S$. Note that $G\brac{V \setminus S}$ is a $(d',\lambda)$-expander, in these models.

\begin{proposition}
	\label{prop:app_a}
	\[ \sum\limits_{i, j \in V \setminus S} A_{ij}\inprod{X_i, X_j} \leq \paren{\lambda k + \dfrac{\paren{d' - \lambda}k^2}{n-k}} \left(1 - \E_{i \sim S} \norm{X_i}^2\right) \mper \]
\end{proposition}

We recall some facts about expander graphs. 
\begin{fact}
	\label{fact:expander_eig}
	Let the eigenvalues of $A_{V \setminus S}$ be $\lambda_1 \geq \lambda \geq \hdots \geq \lambda_{n-k}$ and the corresponding orthonormal eigenvectors be $v_1, v_2, \hdots, v_{n-k}$ then,
	\begin{enumerate}
		\item $\lambda_1 = d'\mper$
		\item $v_1 = \dfrac{\one_{V \setminus S}}{\sqrt{n-k}}\mper$
	\end{enumerate}
\end{fact}

First, we look at an upper bound of the quadratic form with matrix $A_{V \setminus S}$ which is a building block in the proof of our proposition.

\begin{lemma}
	\label{lem:svd_expander}
	For $U \in \R^{n-k}$, we have 
	\[ U^T A_{V \setminus S} U \leq \paren{\dfrac{d'-\lambda}{n-k}} \paren{\sum\limits_{i \in V \setminus S}{U\paren{i}}}^2 + \lambda\norm{U}^2 \mper\]
\end{lemma}
\begin{proof}
	Since $A_{V \setminus S} = \sum_{i \in V \setminus S} \lambda_i v_i v_{i}^T$,
	\begin{align*}
	U^T A_{V \setminus S} U &= U^T \paren{\sum_{i \in V \setminus S} \lambda_i v_i v_{i}^T} U 
	= \sum_{i \in V \setminus S} \lambda_i \inprod{U, v_i}^2\\
	&\leq \paren{d' - \lambda} \inprod{U, v_1}^2 + \lambda \paren{\sum_{i \in V \setminus S} \inprod{U, v_i}^2} \qquad \paren{\text{by \prettyref{fact:expander_eig}}}\\ 
	&= \paren{\dfrac{d' - \lambda}{n-k}} \inprod{U, \one_{V \setminus S}}^2 + \lambda \paren{\sum_{i \in V \setminus S} \inprod{U, v_i}^2}\\
	&= \paren{\dfrac{d' - \lambda}{n-k}} \paren{\sum\limits_{i \in V \setminus S}{U\paren{i}}}^2 + \lambda\norm{U}^2 \qquad \paren{\because v_{i}'s \text{ form a basis}}.
	\end{align*}
\end{proof}

\begin{proof}[Proof of \prettyref{prop:app_a}] 	Similar to the proof of \prettyref{lem:ub_b_ij_xi_xj_s_sbar}, we define $Y$ to be a matrix of size $(n+1) \times (n-k)$ which has $n-k$ column vectors $X_i$, corresponding to each $i \in V \setminus S$.
	\begin{align*}
	\sum\limits_{i, j \in V \setminus S} A_{ij}\inprod{X_{i},X_{j}} &= \sum\limits_{l=1}^{n+1} \paren{Y^{T}_{l}}^T A_{V \setminus S} \paren{Y^{T}_{l}}\\
	&\leq \paren{\dfrac{d' - \lambda}{n-k}} \sum\limits_{l=1}^{n+1}\paren{\sum\limits_{m \in V \setminus S} \paren{Y^{T}_{l}}\paren{m}}^2 \\
	&\qquad+\lambda \sum\limits_{l=1}^{n+1} \norm{Y^{T}_{l}}^2 \qquad \paren{\text{invoking \prettyref{lem:svd_expander} on row vectors of Y}}\\
	&= \paren{\dfrac{d' - \lambda}{n-k}} \paren{\sum\limits_{i, j \in V \setminus S} \inprod{X_i, X_j}} + \lambda \sum\limits_{i \in V \setminus S} \norm{X_i}^2\\ 
	&\qquad\qquad\qquad\qquad\qquad\qquad (\text{rewriting in terms of column vectors})\\
	&\leq \paren{\dfrac{d' - \lambda}{n-k}} \paren{k \sum\limits_{i \in V \setminus S} \norm{X_i}^2} + \lambda \sum\limits_{i \in V \setminus S} \norm{X_i}^2 \qquad (\text{by eqn \prettyref{eq:sdp3}})\\
	&= \paren{\lambda k + \dfrac{\paren{d' - \lambda}k^2}{n-k}} \left(1 - \E_{i \sim S} \norm{X_i}^2\right) \qquad (\text{by eqn \prettyref{eq:sdp2}}).
	\end{align*}
\end{proof}